\documentclass[preprint,pre,twocolumn,10pt,nobibnotes]{revtex4}%
\usepackage{amsfonts}
\usepackage{amsmath}
\usepackage{amssymb}
\usepackage{graphicx}%
\providecommand{\U}[1]{\protect\rule{.1in}{.1in}}
\newtheorem{theorem}{Theorem}

\newtheorem{lemma}[theorem]{Lemma}

\newenvironment{proof}[1][Proof]{\noindent\textbf{#1.} }{\ \rule{0.5em}{0.5em}}
\begin{document}
\preprint{ }
\preprint{UATP/2401}
\title{A No-Go Theorem of Analytical Mechanics for the Second Law Violation \ \ \ }
\author{P.D. Gujrati,$^{1,2}$ }
\affiliation{$^{1}$Department of Physics and $^{2}$School of Polymer Science and Polymer
Engineering, The University of Akron, Akron, OH 44325}
\email{pdg@uakron.edu}

\begin{abstract}
We follow the Boltzmann-Clausius-Maxwell (BCM) proposal (see text) to solve a
long-standing problem of identifying the underlying cause of the second law
(SL) of spontaneous irreversibility, a stochastic universal principle, as the
\emph{mechanical equilibrium} (stable or unstable) \emph{principle} (Mec-EQ-P)
of analytical mechanics as applied to deterministic microstates $\left\{
\mathfrak{m}_{k}\right\}  $ of an isolated nonequilibrium system $\Sigma$ of
\emph{any} size. The principle leads to \emph{nonnegative} system intrinsic
(SI) microwork $\left\{  dW_{k}\right\}  $ and SI-average $dW$ during any
spontaneous process. In conjuction with the first law, Mec-EQ-P leads to a
\emph{generalized second law }(GSL) $dQ=dW\geq0$ for $\Sigma$, where $dQ\doteq
TdS$ is the purely stochastic SI-macroheat that corresponds to $dS\geq0$ for
$T>0$ and $dS<0$ for $T<0$, where $T$\ is the temperature $\left.
dQ/dS\right\vert _{E}$.\ The GSL supercedes the conventional SL formulation
$dS\geq0$ that is valid only for a macroscopic system $\Sigma_{\text{M}}$ for
positive temperatures, but reformulates it to $dS<0$ for negative
temperatures. It is quite surprising that GSL is not only a direct consequence
of intertwined mechanical and stochastic macroquantities through the first law
but also remains valid for any arbitrary irreversible process in $\Sigma$\ of
any size as an identity, allowing $dS\lesseqqgtr0$ for $T\lesseqqgtr0$. It
also becomes a no-go theorem for GSL-violation unless we abandon Mec-EQ-P of
analytical mechanics used in the BCM proposal, which will be catastrophic for
theoretical physics. In addition, Mec-EQ-P also provides new insights into the
roles of spontaneity, nonspontaneity, negative temperatures, instability, and
the significance of $dS\lesseqqgtr0$ due to nonspontaneity and inserting
internal constraints.

\end{abstract}
\date{\today}
\maketitle

Even though Einstein \cite{Einstein} was convinced that classical
thermodynamics\ (Cl-Th) as " ... the only physical theory of universal content
... will never be overthrown," and Eddington \cite{Eddington} considered the
fundamental axiom of the second law (SL) in terms of entropy $S$\ for an
isolated system $\Sigma$\ as its cornerstone holds " ... the supreme position
among the laws of Nature," the recent trend has been to search for
SL-violation
\cite{Evans1993,Evans1994,Evans2002,Gerstner2002,Ebler2018,Ebler2022,Procopio2019,DAbramo2012,Ford2006,Lee,Fu,Moddel,Capek,Pandey}
in $\Sigma$, whose size ranges from mesoscopic to cosmological scales, to cast
doubt on Cl-Th of a nonequilibrium (NEQ) macrostate $\mathfrak{M}$, following
Maxwell \cite{Maxwell}; here, $\mathfrak{M}\doteq\left\{  \mathfrak{m}%
_{k},p_{k}\right\}  $\ is formed by microstates $\mathfrak{m}_{k}$ of the
Hamiltonian $\mathcal{H}$\ with microenergies $E_{k}$ appearing with
probabilities $p_{k},k=1,2,\cdots$. These attempts span widely different
fields from information to biological thermodynamics. Do \emph{positive} and
\emph{negative} temperatures $T$ both satisfy SL as is commonly believed
\cite[for example]{Ramsey}?

There are indirect but strong arguments for a deep connection of SL with
thermodynamic stability \cite{Rovelli2022,
Capela2022,Gavassino2022,Dafermos1979,Callen,
Fosdick1980,KestinV2,Tolman,Landau,Gibbs,Lieb,Gallavotti}, in which a (stable)
NEQ macrostate $\mathfrak{M}^{\text{s}}$ with a lower bound in energy, the
continuous blue curve in Fig. \ref{Fig.1}, asymptotically \emph{converges} to
a unique and stable equilibrium (SEQ) macrostate $\mathfrak{M}_{\text{seq}}$
\cite{Gallavotti} along blue arrows. Thus, a SL-violation strongly suggests,
but not yet verified, thermodynamic \emph{instability} for its cause so we
also consider rarely studied unstable macrostate $\mathfrak{M}^{\text{u}}$,
the continuous red curve, \emph{emerging }out of its unstable equilibrium
(UEQ) macrostate $\mathfrak{M}_{\text{ueq}}$ along red arrows; however, see
\cite{Ramsey} as exception. All arrows point towards \emph{increasing} time
$t$, during which $\mathfrak{M}^{\text{s}}$ (or $\mathfrak{M}^{\text{u}}$)
becomes less (or more) nonuniform with the entropy increasing (decreasing),
with its evolution controlled by its sink $\mathfrak{M}_{\text{seq}}$ (or
source $\mathfrak{M}_{\text{ueq}}$) as succinctly explained in Supplementary
Material \cite{Suppl}. We use compact notation $\mathsf{Q}$\ for
($\mathsf{Q}^{\text{s}}$,$\mathsf{Q}^{\text{u}}$), and $\mathsf{Q}_{\text{eq}%
}$\ for ($\mathsf{Q}_{\text{seq}}$,$\mathsf{Q}_{\text{ueq}}$). Thus, we say
that $\mathfrak{M}_{\text{eq}}$ controls $\mathfrak{M}:$($\mathfrak{M}%
^{\text{s}}$,$\mathfrak{M}^{\text{u}}$)-evolution along continuous curves, and
$\mathfrak{m}_{k\text{eq}}$ controls $\mathfrak{m}_{k}:$($\mathfrak{m}%
_{k}^{\text{s}}$,$\mathfrak{m}_{k}^{\text{u}}$)-evolution along dashed-dot
trajectories. We also let $\mathsf{Q}$ denote $E_{k},E_{k\text{eq}}$, and
$\mathbf{F}_{k}$, the last two defined later.%
\begin{figure}
[ptb]
\begin{center}
\includegraphics[
height=2.2407in,
width=3.5077in
]%
{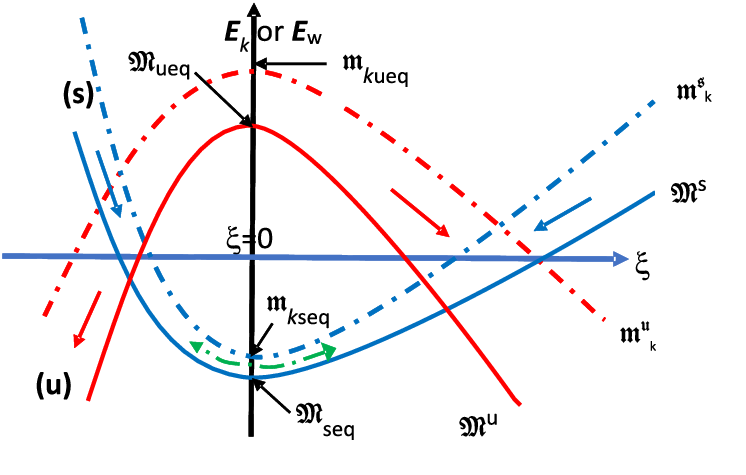}%
\caption{Schematic forms of microenergy $E_{k}$ (dashed-dot curves) and
macrowork function $E_{\text{w}}$ (solid curves) as functions of the internal
variable $\boldsymbol{\xi}$, with $\xi=0$ denoting EQ in $\mathfrak{S}%
_{\mathbf{Z}}$. Alternatively, these curves can be considered as a function of
time $t$ in $\mathfrak{S}_{\mathbf{X}}$, which increases along the directions
of the blue and red arrows. We only consider the case when each curve has a
single extremum. The discussion is easily extended to more complex forms. The
blue color curves and solid blue arrows represent the evolution controlled by
the stable (s) case. The red color curves and solid red arrows represent the
evolution controlled by the unstable (u) case. In both cases, the arrows lower
the energy. The extrema of all curves occur at $\boldsymbol{\xi}=0$ and
represent a uniform body. For the extremum to denote equilibrium, we must also
have $\overset{\cdot}{\boldsymbol{\xi}}=0$ there. The green double-arrow is
discussed in the text. }%
\label{Fig.1}%
\end{center}
\end{figure}

All processes in $\Sigma$\ are \emph{system-intrinsic} (SI) and generated
\emph{internally }to change $E_{k}$. They are commonly called
\emph{spontaneous} \cite{Internal} and are governed by SL, which is only
satisfied by a macroscopic system $\Sigma_{\text{M}}$
\cite{Callen,Landau,Tolman,Lieb}. Any SL-violation, being uncommon and
controversial, will most probably result in puzzling outcomes, not all of
which are recognized or discussed so far in the current literature. Also,
there is no serious inquiry into violation thermodynamics (Viol-Th) supporting
it because the root cause of SL and any special stochasticity
\cite{Note-stochasticity} for its validity are not understood; however, see
\cite{Ramsey}. A \emph{nonspontaneous} process is not internally\emph{
}generated and requires external agent $\Sigma_{\text{ext}}$ so cannot occur
in $\Sigma$.

We prove a generalized SL (GSL), not to be confused with the one proposed by
Beckenstein \cite{Beckenstein,Sewell} for black holes, by considering
$\mathfrak{M}$ in an isolated NEQ discrete\ system $\Sigma$
\cite{Schottky,Muschik} of \emph{any} size, not necessarily $\Sigma_{\text{M}%
}$. It unravels the mystery of and provides deeper insight into the root cause
of SL $dS\geq0$\ and its violation $dS<0$\ \cite{NoEngine}. The equality
occurs in the equilibrium (EQ) macrostate $\mathfrak{M}_{\text{eq}}$ having
maximum entropy \cite{Equilibrium}, and $S$ is always given by Gibbs'
formulation \cite{Gibbs,Shannon} $S\doteq-%
{\textstyle\sum\nolimits_{k}}
p_{k}\ln p_{k}$. There is another system $\Sigma_{\text{E}}$ of intermediate
size determined by the interaction range $\lambda_{\text{E}}$\ to ensure
$E_{k}$-additivity in $\mathfrak{M}$, which is nonuniform in the state space
$\mathfrak{S}_{\mathbf{X}}$, while $\mathfrak{M}_{\text{eq}}$ is uniform as
expressed by $\mathcal{H}$ \cite{Suppl}; here, $\mathbf{X}\doteq
(E,\mathbf{w)}$ is the \emph{fixed} observable set containing energy $E$ and
work parameter $\mathbf{w}$ that uniquely specifies $\mathfrak{M}_{\text{eq}}$
and $\mathcal{H}(\mathbf{w})$. There are two distinct ways to specify
$\mathfrak{M}$. The most common approach is to use $\mathfrak{S}_{\mathbf{X}}$
to nonuniquely specify $\mathfrak{M}$ so that its various quantities such as
$S(\mathbf{X},t)$ or $\mathcal{H}(\mathbf{w},t)$ have \emph{explicit} $t$
dependence. In $\mathfrak{S}_{\mathbf{X}}$, we can focus on a $\Sigma$ of any
size \cite{Suppl}. The alternative is to uniquely specify $\mathfrak{M}%
=\mathfrak{M}(\mathbf{w},\boldsymbol{\xi})=\mathfrak{M}(\boldsymbol{\xi})$ and
$\mathcal{H}(\mathbf{w},\boldsymbol{\xi})=\mathcal{H}(\boldsymbol{\xi})$ in an
extended state space $\mathfrak{S}_{\mathbf{Z}}$ spanned by $\mathbf{Z}%
\doteq(E,\mathbf{W),W}=(\mathbf{w},\boldsymbol{\xi})$ the extended work
parameter including the internal variable $\boldsymbol{\xi}$
\cite{Coleman,Maugin,Guj-entropy,Guj-Foundations,Gujrati-Hierarchy,Langer}
that is determined by internal structures generated by nonuniformity in
$\Sigma_{\text{E}}$ and $\Sigma_{\text{M}}$. It contains $\iota$\ components,
with $\iota$ increasing and $S$ decreasing with nonuniformity of
$\mathfrak{M}(\boldsymbol{\xi})$. There is no explicit time dependence in
$S(\mathbf{Z})$ and $\mathcal{H}(\boldsymbol{\xi})$ in $\mathfrak{S}%
_{\mathbf{Z}}$ that provides some benefit. We define $\boldsymbol{\xi}$\ so
that $\boldsymbol{\xi}$\ $=0$ in $\mathfrak{M}_{\text{eq}}$ \cite[(B)]{Suppl}
for which $\iota=0$. We also assume that $\boldsymbol{\xi}$ is the same for
all $\mathfrak{m}_{k}$'s; extension to different $\boldsymbol{\xi}_{k}$ is
easily done. As explicit time dependence in $\Sigma$ in $\mathfrak{S}%
_{\mathbf{X}}$ is equivalent to the implicit time dependence in $\Sigma
_{\text{E}}$ and $\Sigma_{\text{M}}$ in $\mathfrak{S}_{\mathbf{Z}}$, we can
adopt either approach. We mostly use $\Sigma$ and the first approach\ for its
simplicity. We use $\bar{\Sigma}$ for any system when the state space is not specified.

The four steps below form the core of our analysis: \ 

(S1) Describe $\bar{\Sigma}$ mechanically by its \emph{deterministic}
$\mathcal{H}$\ by specifying $\left\{  \mathfrak{m}_{k},E_{k}\right\}  $ and
SI-microwork $\left\{  dW_{k}\right\}  \doteq-\left\{  dE_{k}\right\}  $.

(S2) Introduce \emph{stochasticity} by appending a probability $p_{k}$\ to
$\mathfrak{m}_{k}$ as was first proposed by Clausius, Maxwell, and Boltzmann
(the BCM\ proposal)
\cite{Maxwell,Clausius,Boltzmann,Tolman,Landau,Gibbs,Gallavotti} to capture
\emph{dissipation} in $\mathfrak{M}$,\ and to identify various ensemble
averages $\left\langle \bullet\right\rangle $ such as its energy $E\doteq%
{\textstyle\sum\nolimits_{k}}
p_{k}E_{k}$, and $S$ above. In general, $S$ is a state function only for
$\Sigma_{\text{M}}$ in $\mathfrak{S}_{\mathbf{Z}}$, but not for $\Sigma$ and
$\Sigma_{\text{E}}$. Different choices for $\left\{  p_{k}\right\}  $ result
in different macrostates and averages for the same $\left\{  \mathfrak{m}%
_{k}\right\}  $. A constant $p_{k},\forall k$, describes a pure mechanical
system \cite{Note-stochasticity} with constant $E$ and $S$.

(S3) Introduce the \emph{first law }$dE=dQ-dW$ \cite[(D)]{Suppl} during any
infinitesimal change of $\bar{\Sigma}$ to determine the allowed change
$\left\{  dp_{k}\right\}  $\ and to identify the physics behind $dQ\doteq%
{\textstyle\sum\nolimits_{k}}
E_{k}dp_{k}$, a stochastic quantity \cite{Note-stochasticity}, and $dW\doteq-%
{\textstyle\sum\nolimits_{k}}
p_{k}dE_{k}=%
{\textstyle\sum\nolimits_{k}}
p_{k}dW_{k}$, a mechanical quantity. The law immediately leads to the
\emph{irreversibility principle} (Irr-P) \cite{Guj-entropy,Gujrati-Hierarchy}
expressed by the identity
\begin{equation}
\emph{\ }dQ=dW\gtreqqless0 \label{IP}%
\end{equation}
that \emph{intertwines} these seemingly unrelated macroquantities similar to
intertwining of electric and magnetic fields in the Maxwell theory. While they
have the same sign, their sign is not fixed yet. We use (S1), (S2), the first
law involving $dQ$ and $dW$, and $dS$ to formulate a \emph{generic} version of
\emph{NEQ statistical thermodynamics} to be called General Thermodynamics
(Gen-Th) for any possible $\mathfrak{M}$ by taking \emph{arbitrary} $\left\{
p_{k}\right\}  $ for $\bar{\Sigma}$; see Eq. (\ref{SecondLaw-General-Entropy}).

(S4) Invoke the mostly overlooked fundamental principle \emph{Mec-EQ-P}
\cite{Arnold,Chetaev} introduced below that controls (spontaneous)
\emph{evolution} of $\mathfrak{m}_{k},\forall k$,\ in $\Sigma$ \cite{Internal}
to formulate the BCM\ proposal for $\mathfrak{m}_{k}$ (MicroBCM) that fixes
the \emph{nonnegative }signature of $dW_{k}$ in Lemma \ref{Lemma}, and to
directly establish the \emph{generalized second law} (GSL) that controls
\emph{dissipation }\cite{Note-stochasticity}\emph{ }in $\bar{\Sigma}$ in
Theorem \ref{Theorem}:%
\begin{equation}
dQ=dW\geq0;dW_{k}\geq0,\forall k.\label{SecondLaw-General-Heat}%
\end{equation}
It requires $\left\{  dp_{k}\right\}  $ to take a special form $\left\{
dp_{k}^{\text{GSL}}\right\}  $ that not only determines $dS$\ but also ensures
$dQ\geq0$; any form different from $\left\{  dp_{k}^{\text{GSL}}\right\}  $
such as $\left\{  -dp_{k}^{\text{GSL}}\right\}  $ does not satisfy GSL and
must be rejected under MicroBCM proposal. The \emph{nonnegative} signature of
$dW$\ after including (S4) in Gen-Th determines the (spontaneous)
\emph{evolution} of $\mathfrak{M}$, and the resulting thermodynamics is
denoted by Gen-GSL-Th. We combine Eqs. (\ref{IP}) and
(\ref{SecondLaw-General-Heat})\ as
\begin{equation}
TdS=dW\left\{
\begin{tabular}
[c]{l}%
$\gtreqqless0$ in Gen-Th\\
$\geq0$ in Gen-GSL-Th
\end{tabular}
\ \ \ \ \right.  ,T\doteq\left.  \frac{dQ}{dS}\right\vert _{E}%
,\label{SecondLaw-General-Entropy}%
\end{equation}
in terms of the SI-temperature $T$ for any $\bar{\Sigma}$, which for
$\Sigma_{\text{M}}$ \cite[(E)]{Suppl} becomes the standard definition $\left(
\partial E/\partial S\right)  _{\boldsymbol{\xi}}$.

The traditional approach in Cl-Th is to supplement (S1), (S2), (S3) by the SL
axiom for $\Sigma_{\text{M}}$ \cite{Callen} instead of (S4) and then prove its
thermodynamic stability \cite{Landau}. However, by \emph{reversing} this
approach, we succeed not only to \emph{prove} GSL for any spontaneous process
in $\bar{\Sigma}$ for any $\left\{  p_{k}\right\}  $ as a \emph{direct
consequence} of $\mathfrak{m}_{k}$-evolution in (S4), but also proclaim a
\emph{No-Go Theorem} of theoretical physics for its violation; see Theorem
\ref{Theorem}. Our proof should be contrasted with the current situation of
proving $dS\geq0$ without ever mentioning $T$; one usually recourses to ad-hoc
assumptions like molecular chaos or master equations
\cite{Note-SecondLaw,vanKampen}. The real significance of our reverse approach
becomes very transparent and unravels many mysteries of SL by recognizing
\emph{fluctuating} $dW_{k}$ as the \emph{primitive }mechanical\emph{ }concept
that completely captures the stochasticity in $dQ$ and the importance of $T$
for SL. Their intertwining provides the basis for the famous
\emph{fluctuation-dissipation theorem} \cite{Callen-Welton,Kubo,FDNote}; see
also \cite{Note-StochasticThermodynamics,Sekimoto,Siefert}. To our knowledge,
our reverse approach with MicroBCM proposal has never been used before to
directly prove SL ($dS\geq0$)\ for $T\geq0$ and its extension requiring a
\emph{reformulation} of $dS<0$ for $T<0$; the latter contradicts the
conventional wisdom \cite{Ramsey}. In addition, GSL also disproves that
instability causes SL violation.

\textsc{Setup: }We describe the salient aspects \cite{Suppl} to prove Theorem
\ref{Theorem} for $\bar{\Sigma}$ in Gen-GSL-Th before justifying them. We
assume \cite{Ramsey1} for simplicity that $E_{k}\ $has only a single extremum
at $E_{k\text{eq}}$ as shown by the dashed-dotted curves in Fig. \ref{Fig.1}
for $\mathfrak{m}_{k}^{\text{s}}$ and $\mathfrak{m}_{k}^{\text{u}}$, each
spontaneously evolving to lower microenergies as $t$ increases along the blue
and red arrows, respectively. From the Uniformity Theorem of $m_{k\text{eq}}$
in \cite[(C)]{Suppl}, the extremum represents a \emph{uniform} microstate
$\mathfrak{m}_{k\text{eq}}$ so the rest of the dashed-dotted curves denote
nonuniform microstates $\mathfrak{m}_{k}$\ ($\mathfrak{m}_{k}^{\text{s}}$ or
$\mathfrak{m}_{k}^{\text{u}}$). Then, $\mathfrak{m}_{k\text{seq}}\ $is the
(asymptotically) \emph{stable} equilibrium of $\mathfrak{m}_{k}^{\text{s}}$,
where $E_{k\text{seq}}$ is \emph{minimum}. Any perturbation away from
$\mathfrak{m}_{k\text{seq}}$ always \emph{restores} $\mathfrak{m}%
_{k}^{\text{s}}$\ back to $\mathfrak{m}_{k\text{seq}}$ so the latter acts as a
sink for which $E_{k}^{\text{s}}\rightarrow E_{k\text{seq}}$ spontaneously as
shown by the blue arrow; we say that $\mathfrak{m}_{k\text{seq}}$ controls
$\mathfrak{m}_{k}^{\text{s}}$-evolution \cite{Kubo} during which
$\mathfrak{m}_{k}^{\text{s}}$ becomes more uniform with $dS>0$. At a
mechanically \emph{unstable} equilibrium $\mathfrak{m}_{k\text{ueq}}$ of
$\mathfrak{m}_{k}^{\text{u}}$, $E_{k\text{ueq}}$\ is \emph{maximum}; any
perturbation from it to $\mathfrak{m}_{k}^{\text{u}}$ spontaneously
\emph{repels} it further away from $\mathfrak{m}_{k\text{ueq}}$. As
$\mathfrak{m}_{k}^{\text{u}}$ \emph{never} returns to $\mathfrak{m}%
_{k\text{ueq}}$ with $E_{k}^{\text{u}}$ running away from $E_{k\text{ueq}}$,
$\mathfrak{m}_{k\text{ueq}}$ becomes a source controlling $\mathfrak{m}%
_{k}^{\text{u}}$-evolution along the red arrow during which $\mathfrak{m}%
_{k}^{\text{u}}$ becomes more nonuniform with $dS<0$. Both are deterministic
spontaneous $\mathfrak{m}_{k}$-evolutions, whose directions are controlled by
the \emph{principle of mechanical equilibrium} (Mec-EQ-P) \cite[(C)]{Suppl}.

In Gen-GSL-Th for $\bar{\Sigma}$, the directional $\mathfrak{m}_{k}$-evolution
performs $dW_{k}=-dE_{k}\geq0$ (Lemma \ref{Lemma}), whose average $dW\geq0$ is
due to spontaneous irreversible processes \cite{Prigogine} and from Irr-P is
dissipated or wasted \cite{Note-stochasticity} as macroheat $dQ$ to ensure
unchanging $E$. Without (S4), Gen-Th also describes GSL-violation ($dQ=dW<0$)
and SL-violation ($dS<0$), resulting in \emph{violation thermodynamics},
Viol-GSL-Th of GSL and Viol-Th of SL, for $\bar{\Sigma}$; see \cite[(C)]%
{Suppl} for various possible thermodynamics.

To make $S$ a \emph{state function} for $\Sigma_{\text{M}}$\ in $\mathfrak{S}%
_{\mathbf{Z}}$ requires an important \emph{restriction} on $\left\{
p_{k}\right\}  $, but not on $\left\{  dp_{k}\right\}  $. The resulting form
of Gen-Th is called \emph{restriction thermodynamics} (Rest-Th)\ for
$\Sigma_{\text{M}}$ \cite[(E)]{Suppl}, which still allows $dS$ of either sign
so $dQ=TdS\gtreqqless0$ in Rest-Th, unless supplemented by (S4); see Eq.
(\ref{SecondLaw-General-Entropy}).

\textsc{Justification: }We justify the above claims for $\Sigma$ in
$\mathfrak{S}_{\mathbf{X}}$\ within Gen-GSL-Th without any restriction on
possible $\left\{  p_{k}\right\}  $. The claims also apply to $\Sigma
_{\text{E}}$ and $\Sigma_{\text{M}}$\ in $\mathfrak{S}_{\mathbf{Z}}$ without
any change. A simple calculation at the end clarifies the claims for a
one-particle system in $\mathfrak{S}_{\mathbf{X}}$.

The extremum and the form of $\left\{  E_{k}\right\}  $ determine the extremum
and the form of $E_{\text{w}}=(E_{\text{w}}^{\text{s}},E_{\text{w}}^{\text{u}%
})$ after averaging \cite[(D)]{Suppl} for $\mathfrak{M}$ as shown by
continuous curves in Fig. \ref{Fig.1}. The averaging also generalizes
microstate Mec-EQ-P to the \emph{thermodynamical principle of (stability and
instability) equilibrium }(Th-EQ-P) for $\mathfrak{M}$ in Gen-GSL-Th,
according to which the extremum $E_{\text{eq}}$ of $E_{\text{w}}$ controls
$\mathfrak{M}$-evolution. We must not confuse this evolution with that
produced by \emph{intervention} required to prepare nonuniform $\mathfrak{M}$
out of uniform $\mathfrak{M}_{\text{eq}}$\ by internal constraints as
discusses below and in \cite{Suppl}; see also\ \cite{Callen,Gujrati-Szilard}.
As $\mathfrak{m}_{k\text{ueq}}\ $is, in effect, physically equivalent to a
"nonexistent" microstate because of its instability, $\mathfrak{M}%
_{\text{ueq}}$ is also physically nonexistent \cite{KestinV2,Wood,Note0}.
Because of this, Cl-Th only deals with $\mathfrak{M}^{\text{s}}$ having the
sink $\mathfrak{M}_{\text{seq}}$ to which it asymptotically converges.
Accordingly any $\mathfrak{m}_{k}^{\text{s}}\in\mathfrak{M}^{\text{s}}$
converges to $\mathfrak{m}_{k\text{seq}}\in\mathfrak{M}_{\text{seq}}$. Despite
this, we also consider $\mathfrak{M}^{\text{u}}$ to obtain additional and
surprising information at the microstate level that is not available in Cl-Th,
and allows for a reformulation of SL for negative $T$.

We first prove the following lemma for an equilibrium point of a mechanical
$\mathfrak{m}_{k}$\ and a thermodynamic $\mathfrak{M}$.

\begin{lemma}
\label{Lemma} During $\mathfrak{m}_{k\text{eq}}$-controlled spontaneous
evolution of $\mathfrak{m}_{k}$ of $\bar{\Sigma}$ in Gen-GSL-Th,
$\mathfrak{m}_{k}$ performs \emph{nonnegative} microwork $\Delta W_{k}$ as
$\mathfrak{m}_{k}\rightarrow\mathfrak{m}_{k}^{^{\prime}}$. Performing ensemble
average with arbitrary $\left\{  p_{k}\right\}  $ then determines
thermodynamic stability (instability) of the resulting macrostate
$\mathfrak{M}=(\mathfrak{M}^{\text{s}},\mathfrak{M}^{\text{u}})$ in
$\bar{\Sigma}$, which spontaneously performs \emph{nonnegative} macrowork
$\Delta W$ in Gen-GSL-Th as $\mathfrak{M}\rightarrow\mathfrak{M}^{^{\prime}}$.
\end{lemma}

\begin{proof}
(a) \textrm{Microstate Evolution}: The \emph{extremum} of $E_{k}$ at
$E_{k\text{eq}}$ represents $\mathfrak{m}_{k\text{eq}}=(\mathfrak{m}%
_{k\text{seq}},\mathfrak{m}_{k\text{ueq}})$ that is uniform in $\mathfrak{S}%
_{\mathbf{X}}$; see Uniformity Theorem in \cite[(C)]{Suppl}. We now consider a
nonuniform $\mathfrak{m}_{k}=(\mathfrak{m}_{k}^{\text{s}},\mathfrak{m}%
_{k}^{\text{u}})$ away from $\mathfrak{m}_{k\text{eq}}$ as it spontaneously
evolves towards $\mathfrak{m}_{k}^{^{\prime}}=(\mathfrak{m}_{k}^{\prime
\text{s}},\mathfrak{m}_{k}^{\prime\text{u}})$, see red and blue arrows in Fig.
\ref{Fig.1}, during which $E_{k}$ spontaneously \emph{decreases }to
$E_{k}^{^{\prime}}$ so that $\Delta E_{k}=-\Delta W_{k}\doteq E_{k}^{^{\prime
}}-E_{k}\leq0$. Thus, $\mathfrak{m}_{k}\ $performs spontaneous
\emph{nonnegative} microwork during its evolution towards $\mathfrak{m}%
_{k}^{^{\prime}}$ given by
\begin{subequations}
\begin{equation}
\Delta W_{k}=-\Delta E_{k}\geq0. \label{deltaE-isolated-micro}%
\end{equation}

(b) \textrm{Macrostate Evolution}: We now average over $\mathfrak{m}%
_{k}^{\text{s}}$ and $\mathfrak{m}_{k}^{\text{u}}$ using arbitrary $p_{k}$ to
obtain macroquantities of $\mathfrak{M}^{\text{s}}$ and $\mathfrak{M}%
^{\text{u}}$, respectively, in Gen-GSL-Th. The mechanical equilibrium
microstate $\mathfrak{m}_{k\text{eq}}$\ determines the thermodynamic EQ
macrostate $\mathfrak{M}_{\text{eq}}=(\mathfrak{M}_{\text{seq}},\mathfrak{M}%
_{\text{ueq}})$.\ The macrowork function $E_{\text{w}}$ gives $E_{\text{eq}%
}\doteq<E_{\text{eq}}>$ for thermodynamic EQ and uniform macrostate
$\mathfrak{M}_{\text{eq}}$ in $\mathfrak{S}_{\mathbf{X}}$ \cite[(C)]{Suppl}.
The spontaneous process $\mathfrak{M}\rightarrow\mathfrak{M}^{^{\prime}}$
results in the nonnegative spontaneous macrowork
\begin{equation}
\Delta W\doteq<\Delta W_{k}>\equiv-<\Delta E_{k}>\text{ }\geq
0.\label{deltaE-isolated-macro}%
\end{equation}

This completes the proof.
\end{subequations}
\end{proof}

We recall that there is no sign restriction on $\Delta W_{k}$ and $\Delta
W$\ in Gen-Th of $\Sigma$. We now prove the main Theorem.

\begin{theorem}
\label{Theorem} We consider $\Sigma$ in $\mathfrak{S}_{\mathbf{X}}$. The
spontaneous $\mathfrak{M}_{\text{eq}}$-evolution of any $\mathfrak{M}$ in
Gen-GSL-Th during which it either converges to the sink $\mathfrak{M}%
_{\text{seq}}$ for $\mathfrak{M}^{\text{s}}$\ or runs away from the
source\ $\mathfrak{M}_{\text{ueq}}$ for $\mathfrak{M}^{\text{u}}$ directly
leads to GSL in Eq. (\ref{SecondLaw-General-Heat}), making it internally
consistent with analytical mechanics. A violation of GSL requires $\Delta W<0$
that cannot occur in $\Sigma$ so the theorem is a No-Go Theorem for GSL
violation. The $\Delta S<0$ due to instability in the spontaneous evolution in
$\mathfrak{M}^{\text{u}}$ is not a violation of SL because of its negative
$T$. This makes Viol-GSL-Th and Viol-Th inconsistent with analytical mechanics
and the BCM proposal.
\end{theorem}

\begin{proof}
In Gen-GSL-Th, the isolated $\Sigma$ in $\mathfrak{M\neq M}_{\text{eq}}$
either spontaneously relaxes as $t$ increases to $\mathfrak{M}_{\text{seq}}$
during which $dS\geq0$ or runs off from $\mathfrak{M}_{\text{ueq}}$ during
which $dS\leq0$. As $\mathfrak{M\rightarrow M}^{\prime}$, we have $\Delta
W\geq0$ from Lemma \ref{Lemma}. It then follows from Eq.
(\ref{SecondLaw-General-Heat}) that the corresponding spontaneous macroheat
$\Delta Q\doteq%
{\textstyle\int\nolimits_{\mathfrak{M}}^{\mathfrak{M}^{\prime}}}
TdS$ is nonnegative, which proves GSL in Eq. (\ref{SecondLaw-General-Entropy})
for any $\mathfrak{M}$ in $\Sigma$, and establishes its consistency with
Mech-EQ-P of analytical mechanics. A GSL-violation $dQ<0$\ can only happen if
$dW<0$ such as along the double green arrow in Fig. \ref{Fig.1} near
$\mathfrak{M}_{\text{seq}}$. As this evolution violates Lemma \ref{Lemma}, it
is \emph{nonspontaneous} and cannot occur in $\Sigma$. This makes any
GSL\ violation impossible and turns the theorem into a No-Go theorem for
GSL-violation. The requirement $dQ\geq0$ is consistent with $T\geq0$ and
$dS\geq0$ for $\mathfrak{M}^{\text{s}}$, and $T<0$ and $dS<0$ for
$\mathfrak{M}^{\text{u}}$. The instability in $\mathfrak{M}^{\text{u}}$ and
$\Delta S<0$\ during its spontaneous evolution does not imply violating SL,
contrary to the comment at the start \cite[for example]{Ramsey}. It is also
clear that Viol-GSL-Th and Viol-Th can only hold for nonspontaneous processes.
As they cannot occur in $\Sigma$, they are inconsistent with analytical
mechanics and the BCM proposal. This completes the proof.
\end{proof}

It is instructive to consider, as an example, a two-level particle as our
system $\Sigma$. The entropy is $S=-p_{1}\ln p_{1}-p_{2}\ln p_{2}$, where
$p_{1}$and $p_{2}$ are the probabilities of the two levels. Then,
$dS=\ln[(1-p_{1})/p_{1}]dp_{1}$. The maximum of $S$ occurs at $p_{1}%
=p_{2}=1/2$, which represents the EQ point. For any other value of $p_{1}$,
the behavior is different for $dS>0$ and $dS<0$. In the former case, $dS$
brings $p_{1}$and $p_{2}$ closer to the EQ point making it a sink, while in
the latter case, $p_{1}$and $p_{2}$ move away from the EQ point making it a
source. This is consistent with our discussion above and brings out the
difference of $\mathfrak{M}^{\text{s}}$ two $\mathfrak{M}^{\text{u}}$ vividly.
To make the example more interesting, we need to associate energies to the two
level. We capture nonuniformity of the system by adding small contributions
$\epsilon_{1}$ and $\epsilon_{2}$ to the equilibrium energies $e_{1}$ and
$e_{2}>e_{1}$, respectively, so that $E_{1}=e_{1}+\epsilon_{1}$ and
$E_{2}=e_{2}+\epsilon_{2}$. We also take $p_{1}=1/2-\delta$ and $p_{2}%
=1/2+\delta$, with $\delta$ a small quantity. We keep only leading order
terms. For $\epsilon_{1}>0$ and $\epsilon_{2}>0$, we are considering
$\mathfrak{M}^{\text{s}}$; see the blue solid curve. Along the blue arrow, we
have $dW\approx(\epsilon_{1}+\epsilon_{2})/2$, and $dQ\approx(e_{1}%
-e_{2})\delta$. From Irr-P, we have $(e_{2}-e_{1})\delta+(\epsilon
_{1}+\epsilon_{2})/2=0$, which is the condition for $dE=0$, as is easily seen
from equating $E_{\text{eq}}=(e_{1}+e_{2})/2$, and $E\approx E_{\text{eq}%
}+(e_{2}-e_{1})\delta+(\epsilon_{1}+\epsilon_{2})/2$, and $\delta<0,dS>0$,
making the system more uniform and $T>0$. For $\epsilon_{1}<0$ and
$\epsilon_{2}<0$, we are considering $\mathfrak{M}^{\text{u}}$; see the red
solid curve. Along the red arrow, we again have $dW>0$ but $dS<0$ as system
becomes more nonuniform. In both cases, $dQ>0$ so GSL remains intact, but the
behavior of $dS$ shows that $T<0$ as expected.

\textsc{Discussion: }We are ready to highlight our findings and discuss their
tantalizing consequences. Gen-Th\ for $\bar{\Sigma}$ always satisfies the
first law $dE=0$ \cite[(D)]{Suppl} (from which follows Irr-P in (S2)) for any
$\left\{  p_{k}\right\}  $ that determines $dW$.\ Irr-P \emph{intertwines}
$dW$ and $dQ$, but without imposing restrictions on $dp_{k}$
\cite{Note-stochasticity} so $dQ$ has either sign but always determined by
$dW$. It is here the relevance of \emph{microscopic} Mec-EQ-P for $\left\{
\mathfrak{m}_{k}\right\}  $, which has not been properly recognized so far,
becomes central in our reverse approach to finally identify Mec-EQ-P\ as the
\emph{root cause} of the stochastic principle of GSL/SL that is surprisingly
and completely determined by \emph{mechanical works }$\left\{  dW_{k}%
\geq0\right\}  $ and $dW$ alone\ in Gen-GSL-Th\ for $\bar{\Sigma}$. Thus,
\emph{GSL is a consequence of analytical mechanics}, the foundation of\emph{
}the BCM proposal, with $dQ\geq0$ directly related to $dS\gtreqqless0$
depending on $T\gtreqqless0$, respectively. This provides a direct proof of SL
for $\mathfrak{M}^{\text{s}}$ for $T\geq0\ $and $dS\geq0$, and extends it to
$\mathfrak{M}^{\text{u}}$ for $T<0\ $and $dS\leq0$, a very tantalizing
extension, which corrects a common \emph{misconception} of SL about negative
$T$ \cite{Ramsey}. As negative temperatures are not physically impossible
\cite{Landau,Ramsey,Purcell,Penrose}, it is quite surprising to realize that
$dS<0$\ is \emph{not} a violation of SL, if $T<0$. This has not been
recognized before; however, see \cite{Campisi}. Thus, Gen-GSL-Th supersedes
Cl-Th for $\Sigma_{\text{M}}$. In addition, Viol-Th and Viol-GSL-Th cannot be
taken as a viable possibility within the BCM proposal.

It is clear from the above theorem that $dS\gtreqqless0$ alone without any
reference to $T$ cannot be used to describe SL or its violation, a point that
has not been appreciated so far and highlights the role of $T$ for SL.

For a genuine GSL/SL violation in $\bar{\Sigma}$, we need nonspontaneous
evolution requiring $T\geq0$ and $dS<0$ for $\mathfrak{M}^{\text{s}}$, and
$T<0$ and $dS>0$ for $\mathfrak{M}^{\text{u}}$. During this evolution,
$\mathfrak{M}^{\text{s}}$ becomes\ less and less uniform ($dS<0$) and
$\mathfrak{M}^{\text{u}}$\ becomes\ more and more uniform ($dS>0$).\emph{
}Thus, $\mathfrak{M}^{\text{s}}$ runs off from the uniform sink $\mathfrak{M}%
_{\text{seq}}$ \emph{catastrophically} to asymptotically approach an extremely
\emph{nonuniform} macrostate $\mathfrak{M}_{\text{cats}}^{\text{s}}$ so that
$S_{\text{cats}}^{\text{s}}<<S_{\text{seq}}$ as if $\mathfrak{M}_{\text{seq}}%
$\ is \emph{unstable}. Similarly, $\mathfrak{M}^{\text{u}}$ runs towards and
terminates in the uniform source $\mathfrak{M}_{\text{ueq}}$\ along the
direction opposite to red arrows with $S$ increasing to $S_{\text{ueq}}$ as if
$\mathfrak{M}_{\text{ueq}}$\ is \emph{stable}. Both possibilities require some
external agent $\Sigma_{\text{ext}}$\ that, for example, inserts an internal
partition for $T>0$ ($T<0$) to obtain $dS<(>)0$ by "\emph{manipulating}
$\mathfrak{M}_{\text{seq}}$ ($\mathfrak{M}_{\text{ueq}}$) " to drive
$\mathfrak{M}^{\text{s}}$ ($\mathfrak{M}^{\text{u}}$)\ run away from (towards)
the sink $\mathfrak{M}_{\text{seq}}$ ($\mathfrak{M}_{\text{ueq}}$) in a
nonspontaneous manner; see \cite[(E)]{Suppl}, where various possibilities are
discussed. If $\bar{\Sigma}$\ is now detached from $\Sigma_{\text{ext}}$, the
result will be an internal constraint discussed by Callen \cite{Callen},
removal of which will initiate spontaneous processes to increase the entropy
in accordance with GSL/SL as it must. A similar situation occurs in the demon
paradox \cite{Knott,Maxwell,Smoluchowski,Gujrati-Maxwell}, in which the demon
starts manipulating $\mathfrak{M}_{\text{seq}}$ in a nonspontaneous manner as
discussed recently \cite{Gujrati-Szilard,Gujrati-Maxwell}.

Th-EQ-P for $\mathfrak{M}^{\text{s}}$ becomes a consequence of SL in Cl-Th for
$\Sigma_{\text{M}}$. Instead, we prove not only Th-EQ-P for $\mathfrak{M}%
^{\text{s}}$ and $\mathfrak{M}^{\text{u}}$ but also GSL for $\bar{\Sigma}$ by
using (S4), \textit{i.e.}, Mec-EQ-P, of analytic mechanics (instead of SL) in
Gen-GSL-Th. This shows the strength of our reverse approach. As $\mathfrak{m}%
_{k}^{\text{s}}$\ and $\mathfrak{M}^{\text{s}}$ asymptotically converge to
$\mathfrak{m}_{k\text{seq}}$\ and $\mathfrak{M}_{\text{seq}}$, respectively,
by removing nonuniformity gradually, $S(t)$ continuously increases and reaches
its \emph{maximum}. As $\mathfrak{m}_{k}^{\text{u}}$\ and $\mathfrak{M}%
^{\text{u}}$ runs away from $\mathfrak{m}_{k\text{ueq}}$\ and $\mathfrak{M}%
_{\text{ueq}}$, respectively, by increasing nonuniformity gradually, $S(t)$
continuously decreases and reaches some \emph{minimum} $S_{\text{cats}}$, even
though GSL remains valid so this case corresponds to $T<0$. With (S4) in
Gen-GSL-Th, BCM proposal imposes a very strong constraint of
\emph{nonnegativity} on $dW_{k}$ for $\bar{\Sigma}$. The nonnegative average
$dW$ through Irr-P determines $dQ$, which imposes a restriction on $\left\{
dp_{k}\right\}  $ to have a particular form $\left\{  dp_{k}^{\text{GSL}%
}\right\}  $ to ensure $dQ\geq0$. Another important finding is that Theorem
\ref{Theorem} also becomes the \emph{No-Go }theorem so that that GSL-violation
requires \emph{rejecting} (S4), \textit{i.e.,} Mec-EQ-P and, therefore, the
BCM proposal, a prospect that has not been recognized so far. Indeed, the
violation cannot be taken as a viable possibility as it defies analytic
mechanics. The conclusion provides not only a tantalizing insight into
SL-violation in $\Sigma_{\text{M}}$ for the first time by recognizing that
Viol-Th requires rejecting analytical mechanics, but is also a testament to
the robustness of the BCM\ proposal. This elevates Mec-EQ-P to be of
\emph{primary} relevance for thermodynamic foundation of GSL/SL in Gen-GSL-Th,
and clarifies the significance of "asymptotic approach to stable\ equilibrium"
for a thermodynamic system in $\mathfrak{M}^{\text{s}}$\ in time after being
momentarily disturbed \cite{Callen,Ruelle,KestinV2,Wood} from $\mathfrak{M}%
_{\text{seq}}$, which forms the cornerstone of Cl-Th. We also extend GSL/SL
now to cover $\mathfrak{M}^{\text{u}}$, where negative $T$ plays an important
role. Thus, SL must always be stated along with $T$, a fact that has not been appreciated.

Thermodynamic behavior of $E_{\text{w}}$ in Fig. \ref{Fig.1} can also be used
to obtain GSL/SL \cite[(D)]{Suppl}, but we have used $E_{k}$ to emphasize
their mechanical nature.

\pagebreak

\begin{center}
{\Large Supplementary Material}
\end{center}

We consider an isolated mechanical system $\Sigma$ containing $N=1,2,\cdots$
particles, which is then treated thermodynamically following the BCM\ proposal
introduced in the main text.\ Its deterministic Hamiltonian $\mathcal{H}%
(\left.  \mathbf{x}\right\vert \mathbf{w})$\ with \emph{no} explicit time
dependence is a function of $\mathbf{x\doteq}\left\{  \mathbf{(r}%
_{i},\mathbf{p}_{i})\right\}  ,i=1,2,\cdots,N$ denoting the set of positions
and momenta of $N$\ particles and specified by the fixed set $\mathbf{w}$ of
extensive observable work parameters (including $N$, volume $V$, etc. but
\emph{not} energy $E$) of $\Sigma$.\ Its microstates $\left\{  \mathfrak{m}%
_{k}(\mathbf{w})\right\}  $\ of energies $\left\{  E_{k}(\mathbf{w})\right\}
$\ are also time-independent and depend only on $\mathbf{w}$\ of the entire
$\Sigma$ and not of its various parts. We express this fact of
\emph{stationarity} by saying that $\mathcal{H}_{\text{eq}}\doteq
\mathcal{H}(\left.  \mathbf{x}\right\vert \mathbf{w}),\mathfrak{m}%
_{k\text{eq}}\doteq\mathfrak{m}_{k}(\mathbf{w})$ and $E_{k\text{eq}}\doteq
E_{k}(\mathbf{w})$\ are \emph{uniform} in $\mathbf{w}$, which is in the spirit
of and consistent with the well-known result of equilibrium (EQ)
thermodynamics \cite{LandauS} that an EQ macrostate $\mathfrak{M}_{\text{eq}%
}\ $is uniform in the set $\mathbf{X}\doteq(E,\mathbf{w})$ of its state
variables that defines the EQ state space $\mathfrak{S}_{\mathbf{X}}$. As a
consequence of uniformity, $\mathfrak{M}_{\text{eq}}\ $and $\mathfrak{m}%
_{k\text{eq}}$\ can neither perform any mechanical work nor generate any
mechanical power. Various \emph{system-intrinsic} (SI) EQ functions including
entropy $S_{\text{eq}}=S(\mathbf{X})$ of $\mathfrak{M}_{\text{eq}}$\ remain
time independent. (From now on, we will mostly suppress $\mathbf{w}$ for
$\Sigma$\ unless necessary as it is fixed.) In a nonequilibrium (NEQ)
macrostate $\mathfrak{M}$ of $\Sigma$ that we are interested in, all
SI-quantities such as $\mathcal{H}(\left.  \mathbf{x}\right\vert t),\left\{
\mathfrak{m}_{k}(t)\right\}  ,\left\{  E_{k}(t)\right\}  $ and $S(t)$ in
$\mathfrak{S}_{\mathbf{X}}$ must change in time $t$ so they acquire an
\emph{explicit} time dependence; see, for example, Eq. (\ref{Micropower-work})
that captures this observation for microwork $dW_{k}(t)$ and micropower
$\mathsf{P}_{k}(t)$ by microstate $\mathfrak{m}_{k}(t)$ \cite{micro-macro}.
This time-dependence is a result of nonuniformity among various disjoined but
mutually interacting uniform subsystems $\left\{  \Sigma_{l}\right\}  $ with
Hamiltonians $\left\{  \mathcal{H}_{l\text{eq}}\doteq\mathcal{H}_{l}(\left.
\mathbf{x}_{l}\right\vert \mathbf{w}_{l})\right\}  $ in terms of work
variables $\mathbf{w}_{l},l=1,2,\cdots$; nonuniformity results in internal
flows among $\left\{  \Sigma_{l}\right\}  $ that affect $\mathbf{X}_{l}%
\doteq(E_{l},\mathbf{w}_{l})$ so a NEQ $\mathfrak{M}(t)$ and $\mathfrak{m}%
_{k}(t)$\ are individually \emph{nonuniform} in $\left\{  \mathbf{X}%
_{l}\right\}  $ and $\left\{  \mathbf{w}_{l}\right\}  $, respectively. For the
isolated system, all SI-changes in it such as $dW_{k}(t)$ and $\mathsf{P}%
_{k}(t)$ are generated internally and are commonly identified as
\emph{spontaneous}. They are governed by the second law (SL) in classical
thermodynamics (Cl-Th) \cite{CallenS}. A \emph{nonspontaneous} process, which
is never under the purview of SL, cannot occur in $\Sigma$. For it to occur,
intervention from an outside agent $\Sigma_{\text{ext}}$ that may also include
Maxwell's demon \cite{MaxwellS} is required to disturb $\Sigma$. The
intervention only creates an internal constraint as discussed by Callen
\cite{CallenS}, whose removal, however, is controlled by SL. \ 

In (B) below, we show that nonuniformity in $\mathfrak{m}_{k}(t)$\ can be
described \emph{precisely} in terms of a set $\boldsymbol{\xi}$ of
\emph{internal variables}
\cite{deGrootS,MauginS,PrigogineS,GujratiS-EntropyReview} if we restrict the
size of subsystems to satisfy \emph{quasi-additivity} of their energies. This
requires imposing size restriction on $\Sigma$, which is denoted by
$\Sigma_{\text{E}}$. We introduce $\mathbf{W}\doteq\left(  \mathbf{w}%
,\boldsymbol{\xi}\right)  $ $\ $as the work variable and $\mathbf{Z}%
\doteq(E,\mathbf{W})$\ as the state variable that determines a NEQ state space
$\mathfrak{S}_{\mathbf{Z}}$ for $\Sigma_{\text{E}}$. The explicit time
dependence of $\mathfrak{M}(t)$ in $\mathfrak{S}_{\mathbf{X}}$ is fully
equivalent to the implicit dependence only through $\boldsymbol{\xi}(t)$ in
$\mathfrak{S}_{\mathbf{Z}}$; see Eq. (\ref{micropower}).

For NEQ statistical mechanics, we follow the BCM proposal and consider an
ensemble in which $\mathfrak{m}_{k}$ appears with probability $p_{k}$, see (D)
below; they together define various macrostates $\mathfrak{M}:\left\{
\mathfrak{m}_{k},p_{k}\right\}  $ for the same set $\left\{  \mathfrak{m}%
_{k}\right\}  $\ but different sets $\left\{  p_{k}\right\}  $, which define
the entropy $S$ by its Gibbs-Shannon formulation \cite{GibbsS} for $\Sigma
$\ as the ensemble average of microentropy $S_{k}\doteq-\eta_{k}$,
\begin{equation}
S=\left\langle S\right\rangle \doteq%
{\textstyle\sum\nolimits_{k}}
p_{k}S_{k},\label{GibbsEntropy}%
\end{equation}
which is easy to verify, see for example \cite{GujratiS-Hierarchy}, for any
arbitrary collection $\left\{  p_{k}\right\}  $. The validity of the above
$S$\ is based on the use of Stirling's approximation. We will not be concerned
about the error in the approximation and accept the formula to be the
definition of $S$. To make the subsystems' entropies additive, we require $S$
to be a \emph{state function} $S(\mathbf{Z})$ in $\mathfrak{S}_{\mathbf{Z}}$,
which requires a further restriction on the size as macroscopic (M) to capture
\emph{quasi-independence} of subsystems. We denote this macroscopic system by
$\Sigma_{\text{M}}$ to be uniquely specified in $\mathfrak{S}_{\mathbf{Z}}$;
in contrast, $\Sigma$\ and $\Sigma_{\text{M}}$ are not uniquely specified in
$\mathfrak{S}_{\mathbf{X}}$ and $\mathfrak{S}_{\mathbf{Z}}$, respectively. We
use $\bar{\Sigma}$\ to denote any of the three systems and any possible
subsystems when we are not interested in specifying the state space.\ 

\begin{center}
.$%
\begin{tabular}
[c]{l}%
Table of $\text{Various Forms of Thermodynamics in the text}$\\
\
\begin{tabular}
[c]{||l||l||}\hline\hline
Gen-Th & (S1),\ (S2), (S3) for $\bar{\Sigma}$, $dW=dQ\gtreqless0$%
\\\hline\hline
Cl-Th & Gen-Th and SL\ axiom for $\Sigma_{\text{M}}$\\\hline\hline
Rest-Th & Gen-Th for $\Sigma_{\text{M}}$\ with state function $S$%
\\\hline\hline
Viol-Th & Gen-Th for $\bar{\Sigma}$\ with $dS<0$\\\hline\hline
Gen-GSL-Th & Gen-Th and (S4) for $\bar{\Sigma}$ with $dQ\geq0$\\\hline\hline
Viol-GSL-Th & Gen-Th and (S4) for $\bar{\Sigma}$ with $dQ<0$\\\hline\hline
\end{tabular}
$\ \ $%
\end{tabular}
$
\end{center}

In NEQ thermodynamics, $\mathfrak{S}_{\mathbf{X}}$ is the most convenient
state space as it requires no knowledge of any internal nonuniformity that is
hidden from the observer and may not always be known, but requires an explicit
time dependence such as in $E_{k}(t)$; see Fig. 1 in main text. Despite this,
the choice of $\mathfrak{S}_{\mathbf{Z}}$ helps to explain not only the origin
of explicit time dependence resulting from nonuniformity but also to justify
the nonstandard definition of $T$ in Eq. (\ref{Temp}). The first law and
macroheat $dQ$, see (D), for $\bar{\Sigma}$ play a central role in proving the
generalized second law (GSL); both can be directly introduced without
identifying $T$. But we need $T$ to relate GSL to $dS$ to also \emph{directly}
prove the second law (SL) $dS\geq0$, and to shed light on the possible
significance of $dS<0$ in its violation, which is one of our goals. In most
cases of interest shown by the solid blue curve $\mathfrak{M}^{\text{s}}$
\cite{LandauS}, $\bar{\Sigma}$ attempts to thermodynamically minimize
nonuniformity in $\left\{  \mathbf{X}_{l}\right\}  $ by reducing internal
flows and become uniform as $S$ continue to increase as explained below. As
$E$ is fixed, we introduce a NEQ temperature as the ratio $dQ/dS$, which we
also simply write as
\begin{equation}
T\doteq\left.  dQ/dS\right\vert _{E}\label{Temp}%
\end{equation}
in any infinitesimal internal process within $\bar{\Sigma}$. While this
definition is unconventional, it is indeed the conventional definition
$\left(  \partial E/\partial S\right)  _{\boldsymbol{\xi}}$ for $\Sigma
_{\text{M}}$\ as seen later from Eqs. (\ref{ArbitraryFirstLaw}) and
(\ref{dQ-dS}). Thus, we take $T$ above to represent the temperature for any
number of particles in $\bar{\Sigma}$.

(A), (B) and (C) below are based on step (S1), see main text, only and deal
with microstates $\left\{  \mathfrak{m}_{k}(t)\right\}  $ of the nonuniform
deterministic Hamiltonian $\mathcal{H}(\left.  \mathbf{x}\right\vert t)$. (A)
deals with it in $\mathfrak{S}_{\mathbf{X}}$\ to show the explicit time
dependence in $\left\{  E_{k}(t)\right\}  $, which is related to the implicit
time-dependence of $\boldsymbol{\xi}(t)$ in $\mathfrak{S}_{\mathbf{Z}}$
introduced in (B). (C) introduces the \emph{Mechanical Equilibrium Principle}
(Mec-EQ-P) of Energy \cite{ArnoldS}, which forms step (S4). This step is
required to \emph{complete} the BCM proposal by connecting it to analytical
mechanics. $\mathfrak{M}(t)$ and their thermodynamics in $\mathfrak{S}%
_{\mathbf{X}}$ and$\ \mathfrak{S}_{\mathbf{Z}}$ are considered next. (D)
requires step (S2) with any arbitrary $p_{k}$ to identify a general NEQ
macrostates $\mathfrak{M}$, and develop its first law for any $\bar{\Sigma}$.
This results in (S3) and the General thermodynamics (Gen-Th) but without
invoking Mec-EQ-P\ for any possible $\mathfrak{M}$. (E) puts macroscopic size
restriction and considers $\Sigma_{\text{M}}$ to ensure $S$ is additive and a
state function in $\mathfrak{S}_{\mathbf{Z}}$.

Various forms of thermodynamics that we consider are listed in the table and
in the main text.

\textsc{(A) Mechanical Systems : }We first consider the mechanical system
$\bar{\Sigma}$ in $\mathfrak{S}_{\mathbf{X}}$ that is specified by
$\mathcal{H}(\left.  \mathbf{x}_{k}\right\vert t)$ given in terms of
$\mathbf{x}_{k}\mathbf{\doteq}\left\{  \mathbf{(r}_{ki},\mathbf{p}%
_{ki})\right\}  $ and with explicit time dependence. Let $\mathfrak{m}%
_{k}(t)\doteq\mathbf{x}_{k}$ denote one of its microstates with microenergy
$E_{k}\doteq\mathcal{H}(\left.  \mathbf{x}_{k}\right\vert t)$ so that
\begin{equation}
dE_{k}=\frac{\partial E_{k}}{\partial\mathbf{x}_{k}}\cdot d\mathbf{x}%
_{k}\mathbf{+}\frac{\partial E_{k}}{\partial t}dt=\frac{\partial E_{k}%
}{\partial t}dt \label{MicroenergyChange0}%
\end{equation}
as the first term vanishes due to Hamilton's equations of motion. As this term
includes the effect of interparticle potentials, we are not interested in
those potentials so we will no longer exhibit $\mathbf{x}_{k}$, which is
explicitly identified by the suffix $k$. Instead, we are interested in
$dE_{k}$\ that is determined solely by varying the parameter $t$
\cite{LandauS-Mech}. We write $\mathcal{H}(\left.  \mathbf{x}_{k}\right\vert
t)$ as $E_{k}(t)$ or simply $E_{k}$; similarly, we simply write $\mathfrak{m}%
_{k}$\ for $\mathfrak{m}_{k}(t)$ unless clarity is needed. The SI-microwork
$dW_{k}(t)$ can be expressed in terms of instantaneous power $\mathsf{P}%
_{k}(t)$ in $\mathfrak{S}_{\mathbf{X}}$\ as
\begin{equation}
dW_{k}(t)=-dE_{k}(t)=\mathsf{P}_{k}(t)dt. \label{Micropower-work}%
\end{equation}

Similarly, the mechanical system $\Sigma_{\text{E}}$ in a microstate
$\mathfrak{m}_{k}(\boldsymbol{\xi})$\ in $\mathfrak{S}_{\mathbf{Z}}$\ is
specified by $\mathcal{H}(\left.  \mathbf{x}_{k}\right\vert \boldsymbol{\xi}%
)$, which now depends on the set $\boldsymbol{\xi}$ of work parameters of
$\Sigma_{\text{E}}$, see (B) below, that determines the microwork
$dW_{k}(\boldsymbol{\xi})$ as we now explain. We remark that $\boldsymbol{\xi
}$ is assumed to be the same for all microstates so it does not carry the
suffix $k$, while the microforce $\mathbf{F}_{k}$ in Eq. (\ref{GenMicroforce})
and $dW_{k}$\ do so the latter are \emph{fluctuating} (over $\mathfrak{m}_{k}%
$). We will not consider here the case when $\boldsymbol{\xi}$ is fluctuating
(over $\mathfrak{m}_{k}$), to which the present method is easily extended
\cite{GujratiS-LangevinEq,GujratiS-Foundations}. The significance of
$\boldsymbol{\xi}$ becomes clear when we consider the microenergy
$E_{k}(\boldsymbol{\xi})\equiv\mathcal{H}(\left.  \mathbf{x}_{k}\right\vert
\boldsymbol{\xi})$ of $\mathfrak{m}_{k}(\boldsymbol{\xi})$, for which%
\begin{equation}
dE_{k}=\frac{\partial E_{k}}{\partial\boldsymbol{\xi}}\cdot d\boldsymbol{\xi
}\mathbf{\ } \label{MicroenergyChange}%
\end{equation}
by using the same argument as used in Eq. (\ref{MicroenergyChange0}). As
above, we simply denote $E_{k}(\boldsymbol{\xi})$ by $E_{k}$, and
$\mathfrak{m}_{k}(\boldsymbol{\xi})$ by $\mathfrak{m}_{k}$\ unless clarity is
needed. The SI-microforce, also known as the generalized microforce, in
$\mathfrak{m}_{k}$ is given by
\begin{equation}
\ \mathbf{F}_{k}\doteq-\partial E_{k}/\partial\boldsymbol{\xi},
\label{GenMicroforce}%
\end{equation}
and the SI-microwork done by it is given by%
\begin{equation}
dW_{k}(\boldsymbol{\xi})=\mathbf{F}_{k}\cdot d\boldsymbol{\xi}=-dE_{k}%
\mathbf{,} \label{GenMicrowork}%
\end{equation}
which clarifies the significance of $\boldsymbol{\xi}$ as the \emph{work
parameter} for $\Sigma_{\text{E}}$. In NEQ thermodynamics, the macroscopic
analog (ensemble average $\left\langle \mathbf{\bullet}\right\rangle $, see
(D) below) $\mathbf{F}=\left\langle \mathbf{F}\right\rangle $ of
$\mathbf{F}_{k}$\ determines the \emph{affinity} \cite{deGrootS,PrigogineS}
associated with $\boldsymbol{\xi}$. For us, $\boldsymbol{\xi}$ describes
nonuniformity and changes spontaneously with time in accordance with the
\emph{principle of mechanical equilibrium}, called Mec-EQ-P; see (C) below,
the main text and \cite[p. 99]{ArnoldS}. Thus, $E_{k}(\boldsymbol{\xi})$
continues to change with $\boldsymbol{\xi}$ and results in the above
microwork. The corresponding instantaneous microstate power in $\mathfrak{S}%
_{\mathbf{Z}}$\ is
\begin{equation}
\mathsf{P}_{k}(t)\doteq\mathbf{F}_{k}\cdot\boldsymbol{\dot{\xi}},
\label{micropower}%
\end{equation}
which can be used in Eq. (\ref{Micropower-work}) to obtain $dW_{k}(t)$. We
thus see that implicit time-dependence in $\boldsymbol{\xi}$ for
$\Sigma_{\text{E}}$\ in $\mathfrak{S}_{\mathbf{Z}}$\ is equivalent to the
explicit time dependence in $\mathfrak{S}_{\mathbf{X}}$ without any
$\boldsymbol{\xi}$ for $\Sigma$. The variation of $E_{k}(\boldsymbol{\xi})$
for an isolated system in $\mathfrak{S}_{\mathbf{Z}}$ as shown in Fig. 1 in
the main text with $\boldsymbol{\xi}$ can also be treated as time variation of
$E_{k}(t)$ in $\mathfrak{S}_{\mathbf{X}}$. This is also consistent with the
way a time-dependent Hamiltonian is treated in mechanics \cite{LandauS-Mech}
by considering $t$ as a parameter in the Hamiltonian so we can use $t$ instead
of $\boldsymbol{\xi}$,\ in which case Eq. (\ref{GenMicrowork}) reduces to Eq.
(\ref{Micropower-work}) obtained in $\mathfrak{S}_{\mathbf{X}}$, where
$\Sigma$ has no restriction on its size. We will not be concerned with the
actual time dependence of $\boldsymbol{\xi}$ in $\mathfrak{S}_{\mathbf{Z}}$ in
this investigation; all we need to remember that it can be accounted for by an
explicit time dependence in $\mathsf{P}_{k}(t)$ and other quantities for
$\Sigma$ in $\mathfrak{S}_{\mathbf{X}}$.

According to the discussion above, we can consider either $\boldsymbol{\xi}$
explicitly (to be specified in (B) below) and use the state space
$\mathfrak{S}_{\mathbf{Z}}$, or consider $t$ instead in the state space
$\mathfrak{S}_{\mathbf{X}}$ without any specification of $\boldsymbol{\xi}$.
In the former case, microstates are specified in $\mathfrak{S}_{\mathbf{Z}}$,
which restricts $\Sigma$ to $\Sigma_{\text{E}}$. NEQ\ temperature in Eq.
(\ref{Temp}) is always defined.

\textsc{(B) Internal Variables for Mechanical microstates }$\left\{
\mathfrak{m}_{k}\right\}  $ \textsc{of} $\Sigma_{\text{E}}$\textsc{: }We now
turn to the definition of internal variable
\cite{deGrootS,MauginS,PrigogineS,GujratiS-EntropyReview}\ for a deterministic
nonuniform mechanical system $\Sigma_{\text{E}}$, but denote it by $\Sigma$ to
simplify notation unless clarity is needed. For simplicity, we consider it to
be formed by two different and disjoined but mutually interacting subsystems
$\Sigma_{1}$ and $\Sigma_{2}$ that are each uniform in $\mathfrak{S}%
_{\mathbf{X}}$. The internal variables are required to \emph{uniquely} specify
$\left\{  \mathfrak{m}_{k}\right\}  $ in $\mathfrak{S}_{\mathbf{Z}}$; they are
not specified uniquely in $\mathfrak{S}_{\mathbf{X}}$. We use suffixes for
$\Sigma_{1},\Sigma_{2}$ and no suffix for $\Sigma$. We consider $\Sigma
,\Sigma_{1}$, and $\Sigma_{2}$ in microstate $\mathfrak{m}_{k},\mathfrak{m}%
_{k_{1}}$, and $\mathfrak{m}_{k_{2}}$ of energy $E_{k},E_{k_{1}}$, and
$E_{k_{2}}$, respectively. We also introduce $n_{1}=N_{1}/N$ and $n_{2}%
=N_{2}/N$. If we \emph{neglect} the energy $\delta E_{k}$ due to the interface
between $\Sigma_{1}$ and $\Sigma_{2}$, the Hamiltonians of the three bodies
(we do not show them their arguments $\mathbf{x}_{k},\mathbf{x}_{k_{1}}$ and
$\mathbf{x}_{k_{2}}$, respectively, but use microstate suffixes) are related
by
\begin{subequations}
\begin{equation}
\mathcal{H}_{k}(\mathbf{w},\boldsymbol{\xi})\approxeq\mathcal{H}_{k_{1}%
}(\mathbf{w}_{1})+\mathcal{H}_{k_{2}}(\mathbf{w}_{2}),
\label{Hamiltonian-Isolated}%
\end{equation}
which ensures microenergy \emph{quasi-additivity}
\begin{equation}
E_{k}\approxeq E_{1k_{1}}+E_{2k_{2}}, \label{Energy-Isolated}%
\end{equation}
which can be justified only if we \emph{restrict} the minimum sizes of
$\Sigma_{1}$ and $\Sigma_{2}$ to be some $\lambda_{\text{E}}$ that itself is
determined by the range of inter-particle interactions. Under this conditions,
we denote $\Sigma,\Sigma_{1}$, and $\Sigma_{2}$ by $\Sigma_{\text{E}}%
,\Sigma_{1\text{E}}$, and $\Sigma_{2\text{E}}$, and we have $\mathfrak{m}%
_{k}\approxeq\mathfrak{m}_{k_{1}}\otimes\mathfrak{m}_{k_{2}}$. However, for
the sake of notational simplicity, we suppress E in the suffix for various
quantities except for bodies.

The discussion in (A) applies to any $\Sigma_{\text{E}}$ and $\Sigma
_{\text{M}}$. In general, $\Sigma_{\text{E}}$\ is much smaller in size than
$\Sigma_{\text{M}}$. The additivity of extensive (which $t$ is not)\ work
parameters $\mathbf{w}_{1}$ and $\mathbf{w}_{2}$\ in Eq. (\ref{W-Sum}) is not
affected by the presence or absence of $\delta E_{k}$ so that we can introduce
the extensive work parameter set
\begin{equation}
\mathbf{w}\doteq\mathbf{w}_{1}+\mathbf{w}_{2}. \label{W-Sum}%
\end{equation}
We introduce a new extensive \emph{internal variable} set $\boldsymbol{\xi
}_{k}\doteq(\xi_{k\text{E}},\boldsymbol{\xi})$, where%
\end{subequations}
\begin{equation}
\xi_{k\text{E}}\doteq E_{1k_{1}}/n_{1}-E_{2k_{2}}/n_{2},\boldsymbol{\xi}%
\doteq\mathbf{w}_{1}/n_{1}-\mathbf{w}_{2}/n_{2},
\label{IndependentCombinations0}%
\end{equation}
and the extensive work parameter%
\begin{equation}
\mathbf{W}\doteq(\mathbf{w,}\boldsymbol{\xi}\mathbf{)}
\label{Isolated-WorkParameter}%
\end{equation}
are independent variables for $\Sigma_{\text{E}}$ in that $\left(  E_{k}%
,\xi_{k\text{E}}\right)  $ and $\left(  w_{l}\in\mathbf{w},\xi_{l}%
\in\boldsymbol{\xi}\right)  $ form pairs of independent variables. The
introduction of $\boldsymbol{\xi}_{k}$ is to ensure that the number of
variables on both sides in Eq. (\ref{Hamiltonian-Isolated}) are equal as they
must be for an equality. We also have%
\begin{equation}
E_{1k_{1},2k_{2}}=n_{1,2}(E_{k}\pm n_{2,1}\xi_{k\text{E}}),\mathbf{w}%
_{1,2}=n_{1,2}(\mathbf{w}\pm n_{2,1}\boldsymbol{\xi}). \label{Combinations}%
\end{equation}
\qquad We can easily extend the above discussion to any numbers $m$ of
subsystems $\left\{  \Sigma_{l}\right\}  ,l=1,2,\cdots,m$ forming
$\Sigma_{\text{E}}$, each specified by its own observable set $\mathbf{w}_{l}$
to allow for a complex form of nonuniformity in terms of uniform subsystems;
Eq. (\ref{Isolated-WorkParameter}) remains valid for any $m$. It is easy to
verify that all internal variables in $\boldsymbol{\xi}$, whose number we
denote by $\iota$, can be expressed as a \emph{linear combinations} of
$\left\{  \mathbf{w}_{l}\right\}  $
\cite{GujratiS-Foundations,GujratiS-EntropyReview}, with $m$ and $\iota$
increasing with the degree of nonuniformity of $\mathfrak{m}_{k}%
(\boldsymbol{\xi})$. Even though the definition of $\boldsymbol{\xi}_{k}$\ is
not unique, we choose it for convenience so that it vanishes when
$\Sigma_{\text{E}}$\ is uniform \textit{i.e.}, in EQ as defined in Eq.
(\ref{IndependentCombinations0}). The set $\left(  E_{k},\boldsymbol{\xi}%
_{k}\right)  $ forms the complete set of variables to uniquely specify
$\mathfrak{m}_{k}$ of $\Sigma_{\text{E}}$ in $\mathfrak{S}_{\mathbf{Z}}$. It
is easily verified that
\begin{equation}
F_{k\text{E}}\doteq-\partial E_{k}/\partial\xi_{k\text{E}}=0, \label{Work-kE}%
\end{equation}
a consequence of independent $E_{k}$ and $\xi_{k\text{E}}$ so the variation of
$\xi_{k\text{E}}$ does not generate any SI-microwork for any $\left\{
\mathbf{w}_{m}\right\}  $. This explains why it is not included in
$\mathbf{W}$, in which $\mathbf{w}$ remains constant as $\Sigma$\ evolves in
time; however, $\boldsymbol{\xi}_{k}$ continues to change due to internal flows.

Here, we focus on two ($m=2$) subsystems $\Sigma_{1\text{E}}$\ and
$\Sigma_{2\text{E}}$ of $\Sigma_{\text{E}}$, for simplicity, each of which is
uniform (no internal variables for them) but $\Sigma_{\text{E}}$ is not. We
let $E_{k_{l}},\mathbf{w}_{l}\doteq(N_{l},V_{l})$ specify uniform
microstates$\ \mathfrak{m}_{k_{l}},$ $l=1,2$ in $\mathfrak{S}_{\mathbf{X}}$.
We now focus on the nonuniform microstate $\mathfrak{m}_{k}$ of $\Sigma
_{\text{E}}$, which require internal variables along with $E_{k}$\ and
$\mathbf{w}\doteq(N,V)$\ to be specified in $\mathfrak{S}_{\mathbf{Z}}$. We
keep observables $E_{k}$\ and $\mathbf{w}\doteq(N,V)$\ of $\Sigma_{\text{E}}%
$\ fixed along with $N_{1}$ and $N_{2}$ of $\Sigma_{1\text{E}}$\ and
$\Sigma_{2\text{E}}$. Following Eq. (\ref{IndependentCombinations0}), we have
\begin{equation}
\xi_{k\text{E}}\doteq E_{k_{1}}/n_{1}-E_{k_{2}}/n_{2},\xi_{\text{V}}\doteq
V_{1}/n_{1}-V_{2}/n_{2}, \label{IndependentCombinations2}%
\end{equation}
to identify $\mathbf{Z}=(E_{k},N,V,\xi_{k\text{E}},\xi_{\text{V}})$ for
$\Sigma_{\text{E}}$.

Densities in $\Sigma_{1\text{E}}$ and $\Sigma_{2\text{E}}$ are equal in a
\emph{uniform }microstate $\mathfrak{m}_{k}(\boldsymbol{\xi}_{k})$ of
$\Sigma_{\text{E}}$ so $\boldsymbol{\xi}_{k}=0$ and need not be considered.
This is consistent with the fact that uniform $\mathfrak{m}_{k\text{eq}}$ is
uniquely specified in $\mathfrak{S}_{\mathbf{X}}$. In this case, we have a
trivial additivity of the Hamiltonians in Eq. (\ref{Hamiltonian-Isolated})
given by
\begin{equation}
\mathcal{H}_{k}(\mathbf{w})\simeq\mathcal{H}_{k_{1}}(\mathbf{w}_{1}%
)+\mathcal{H}_{k_{2}}(\mathbf{w}_{2}) \label{Hamiltonian-Isolated-Uniform}%
\end{equation}
with no internal flows between subsystems. In a nonuniform microstate,
$\boldsymbol{\xi}_{k}\neq0$. We recall that $\mathbf{w}$ is \emph{fixed} for
$\mathfrak{m}_{k}(\boldsymbol{\xi}_{k})$ but $\mathbf{w}_{1}$ and
$\mathbf{w}_{2}$ can change due to possible transfers (internal flows) between
$\Sigma_{1\text{E}}$ and $\Sigma_{2\text{E}}$ with $d\mathbf{w}_{1}%
=-d\mathbf{w}_{2}$. It follows from Eq. (\ref{Combinations}) that%
\begin{equation}
\dot{E}_{1k_{1},2k_{2}}=n_{1,2}(\dot{E}_{k}\pm n_{2,1}\dot{\xi}_{k\text{E}%
}),\mathbf{\dot{w}}_{1}=-\mathbf{\dot{w}}_{2}=n_{1}n_{2}\boldsymbol{\dot{\xi}}
\label{CombinatioRate}%
\end{equation}
for $\mathfrak{m}_{k}(\boldsymbol{\xi}_{k})$.

As $\mathbf{W}$, \textit{i.e.}, $\boldsymbol{\xi}$\ in Eq.
(\ref{Isolated-WorkParameter}) for $\Sigma_{\text{E}}$ is the work variable in
$\mathcal{H}_{k}(\boldsymbol{\xi}_{k})$ (we suppress $\mathbf{w}$ as it it
constant), $E_{k}(\boldsymbol{\xi}_{k})$ corresponding to $\mathfrak{m}%
_{k}(\boldsymbol{\xi}_{k})$ is only a function of $\boldsymbol{\xi}$\ due to
Eq. (\ref{Work-kE}). This is shown in Fig. 1 in main text. As $\boldsymbol{\xi
}_{k}$\ continuously changes due to transfers between $\Sigma_{1\text{E}}$ and
$\Sigma_{2\text{E}}$, this causes variations in $E_{k}(\boldsymbol{\xi})$. The
variation is similar to the variation in $E_{k}(t)$ for a $\Sigma$ in
$\mathfrak{S}_{\mathbf{X}}$ as discussed in (A).

In summary, as long as an internal variable is used, the system is restricted
to be at least $\Sigma_{\text{E}}$ in size and requires $\mathfrak{S}%
_{\mathbf{Z}}$ for unique specification. There is no explicit time dependence
in this state space. If considerations of subsystems are not important in any
discussion then quasi-additivity is not an issues as is the case for the the
discussion of Eq. (\ref{Micropower-work}). In that case, $\Sigma$\ can be
considered as a whole in $\mathfrak{S}_{\mathbf{X}}$; there is no need to
consider $\delta E_{k}$ separately, which is included in $\mathcal{H}%
_{k}(\mathbf{w},t)$. Thus, the mechanical situation here is that of (A).

\textsc{(C) Mechanical Equilibrium Principle of Energy} \textsc{(Mec-EQ-P) for
$\mathfrak{m}_{k}$: }Time-independent microstates $\left\{  \mathfrak{m}%
_{k\text{eq}}\right\}  $ of a Hamiltonian system, being uniform as
demonstrated by the Uniformity Theorem\emph{ }below, play an important role in
our approach as we explain now. Any nonuniformity in $\bar{\Sigma}$ endows
$\left\{  \mathfrak{m}_{k}\right\}  $ with explicit time dependence as
$\left\{  \mathfrak{m}_{k}(t)\right\}  $ for $\Sigma$ or implicit time
dependence through $\boldsymbol{\xi}_{k}$\ as $\mathfrak{m}_{k}%
(\boldsymbol{\xi}_{k})$ for $\Sigma_{\text{E}}$ or $\Sigma_{\text{M}}$ as
shown by $\left\{  \mathfrak{m}_{k}^{\text{s}}\right\}  $\ and $\left\{
\mathfrak{m}_{k}^{\text{u}}\right\}  $; see dashed-dotted blue and red curves,
respectively; the directions of blue and red arrows on them denote increasing
$t$ during their temporal SI-evolution controlled by internally generated
processes that are mechanically spontaneous and generate micropower
$\mathsf{P}_{k}(t)$. In analytical mechanics, there is no dissipation in
evolution but $\mathfrak{m}_{k}^{\text{s}}$ is special in that its evolution
will normally undergo oscillations about $\mathfrak{m}_{k\text{seq}}$ that
will persist forever as manifested by the dashed-dot blue curve in . 1 in main
text; in contrast, the evolution of $\mathfrak{m}_{k}^{\text{u}}$ has no
oscillation, see the dashed-dot red curve there, as $\mathfrak{m}%
_{k}^{\text{u}}$\ runs away from $\mathfrak{m}_{k\text{ueq}}$ but terminates
is a \emph{catastrophe} in which $\mathfrak{m}_{k}^{\text{u}}$ becomes
extremely nonuniform, which for $\Sigma_{\text{E}}$ or $\Sigma_{\text{M}}$
corresponds to an extremely large $\iota=\iota_{\text{cats}}$; see (E) below.
Thus, $\mathfrak{m}_{k\text{ueq}}$ is the \emph{source} for the SI-evolution
of $\mathfrak{m}_{k}^{\text{u}}$.

As we are eventually interested in $\mathfrak{m}_{k}^{\text{s}}$\ and
$\mathfrak{m}_{k}^{\text{u}}$ in a thermodynamic setting, there will be
macroheat to be be simultaneously considered that we discuss in (D) below. As
is well known, its presence gives rise to dissipation so all stable
thermodynamic processes in $\mathfrak{M}^{\text{s}}$ terminate asymptotically
at $\mathfrak{M}_{\text{seq}}$\ as shown by the solid blue curve in the
figure. Keeping this in hindsight, we intentionally overlook oscillations in
$\mathfrak{m}_{k}^{\text{s}}$ as they do not affect the change in
$E_{k}^{\text{s}}$ during its SI-evolution to $\mathfrak{m}_{k\text{seq}}$
only, making the latter as the \emph{sink} for the SI-evolution of
$\mathfrak{m}_{k}^{\text{s}}$. This is a useful strategy, since oscillations
become thermodynamically irrelevant as $\mathfrak{M}^{\text{s}}$
asymptotically approaches $\mathfrak{M}_{\text{seq}}$; see below. This also
clarifies the importance of the uniform microstate $\mathfrak{m}_{k\text{eq}}%
$; see the UniformityTheorem below.

We first treat $\Sigma_{\text{E}}$, whose results pave the way for a clear
understanding of what to expect for $\Sigma$. In analytical mechanics, the
temporal evolution of $\boldsymbol{\xi}_{k}$ and $\mathfrak{m}_{k}%
(\boldsymbol{\xi}_{k})$\ in $\mathfrak{S}_{\mathbf{Z}}$\ is controlled by
$\mathfrak{m}_{k\text{eq}}$, and is governed by the emergent SI-microforce
$\mathbf{F}_{k}$ in Eq. (\ref{GenMicroforce}) by controlling internal flows
within $\mathfrak{m}_{k}(\boldsymbol{\xi}_{k})$. The emergent processes in
$\Sigma_{\text{E}}$ resulting in the SI-microwork in Eq. (\ref{GenMicrowork})
are commonly identified as \emph{spontaneous} since they are internally
controlled by SI-microforce $\mathbf{F}_{k}$\ and not from the outside by any
nonsystem microforce. The mechanical equilibrium (Mec-EQ) point at
$E_{k}^{\text{eq}}$ is the \emph{equilibrium point} \cite{ArnoldS} in the
SI-evolution under the SI-micorforce $\mathbf{F}_{k}$, where not only the
SI-internal velocity $\boldsymbol{\dot{\xi}}_{k}^{\text{eq}}$ but also the
SI-micorforce $\mathbf{F}_{k}^{\text{eq}}$ \emph{vanish}:
\begin{subequations}
\label{CriticalPoint}%
\begin{align}
\boldsymbol{\dot{\xi}}_{k}^{\text{eq}}  &  =0,\label{IsolatedEQ}\\
\mathbf{F}_{k}^{\text{eq}}  &  =-\partial E_{k}/\partial\boldsymbol{\xi}%
_{k}^{\text{eq}}=0. \label{Microforce-Isolated}%
\end{align}

It follows from Eq. (\ref{Microforce-Isolated}) that this point represents an
extremum in $E_{k}$; see Fig. 1 in main text. Its \emph{minimum} at
$E_{k}^{\text{seq}}=E_{k}^{\text{eq}}$ representing a mechanically\emph{
stable EQ} point enunciates the \emph{mechanical asymptotic stability
principle of minimum energy}. In thermodynamics, stability refers to
asymptotic stability in that the system must eventually approach stable
equilibrium so that $\boldsymbol{\xi}_{k}(t)\rightarrow0$ as $t\rightarrow
\infty$ \cite{CallenS}. In this sense, we see that the absolute value
$\left\vert \boldsymbol{\xi}_{k}(t)\right\vert $ behaves similar to $1/t$ in
the figure. The asymptotic approach is a stronger requirement than just
imposing stability in which the system never strays far away from stable EQ
point. In contrast, the \emph{maximum of }$E_{k}$ at $E_{k}^{\text{ueq}}%
=E_{k}^{\text{eq}}$ representing a mechanically \emph{unstable EQ} point
enunciates the \emph{mechanical instability principle of maximum energy}. In
this case, $\left\vert \boldsymbol{\xi}_{k}(t)\right\vert $ behaves similar to
$t$ in the figure. Both kinds of points are determined by the curvature of
$E_{k}$ at $E_{k}^{\text{eq}}$. The two principles are collectively called
\emph{Mec-EQ-P} in this investigation.

We now prove the following important theorem emphasizing the physical
significance of the uniformity of the Mec-EQ point $\mathfrak{m}_{k\text{eq}}$
for thermodynamic stability, which must exist only in $\mathfrak{S}%
_{\mathbf{X}}$ at $\boldsymbol{\xi}_{k}=0$, where $\mathfrak{M}_{\text{eq}}$ exists.
\end{subequations}
\begin{theorem}
\emph{Uniformity Theorem of }$\mathfrak{m}_{k\text{eq}}$ $:$\emph{ }The Mec-EQ
point $\mathfrak{m}_{k\text{eq}}$ of $\mathfrak{m}_{k}$ is stationary and
uniform microstate with no internal flows between the microstates of its
uniform subsystems.
\end{theorem}

\begin{proof}
We begin by considering a microstate $\mathfrak{m}_{k}(t)$ for $\Sigma$\ in
$\mathfrak{S}_{\mathbf{X}}$ to establish that $\mathfrak{m}_{k\text{eq}%
}:(\mathfrak{m}_{k\text{seq}},\mathfrak{m}_{k\text{ueq}})$ is stationary. The
asymptotic convergence of $\mathfrak{m}_{k}^{\text{s}}$ to $\mathfrak{m}%
_{k\text{seq}}$ as $t\rightarrow\infty$ means that $\mathfrak{m}_{k\text{seq}%
}$ must be \emph{stationary} so it must be a uniform microstate of the
time-independent Hamiltonian $\mathcal{H}(\mathbf{x})$. The instability of
$\mathfrak{m}_{k\text{ueq}}$\ is due to the instability of the $\mathcal{H}%
(\mathbf{x})$ in this case so it\ is also uniform.

As $\Sigma$ contains no information about its internal structure, it cannot be
used to understand internal flows for which we need to at least consider
$\Sigma_{\text{E}}$ to justify the remainder part of the theorem. From Eq.
(\ref{IsolatedEQ}), we obtain $\dot{\xi}_{k\text{E}}^{\text{eq}}=0$ and
$\boldsymbol{\dot{\xi}}^{\text{eq}}=0$. Using first $\boldsymbol{\dot{\xi}%
}^{\text{eq}}=0$ in Eq. (\ref{CombinatioRate}), we conclude that $\dot{E}%
_{k}^{\text{eq}}=0$, which follows directly from Eq.
(\ref{Hamiltonian-Isolated-Uniform}). Using then $\dot{\xi}_{k\text{E}%
}^{\text{eq}}=0$ in Eq. (\ref{CombinatioRate}) along with $\dot{E}%
_{k}^{\text{eq}}=0$, we conclude that $\dot{E}_{1k_{1},2k_{2}}^{\text{eq}}=0$.
Together, they show that \emph{all} flows cease at $E_{k}^{\text{eq}}$, which
makes $\mathfrak{m}_{k\text{eq}}$ \emph{stationary} and, therefore,
\emph{uniform} as above. Any nonzero $\boldsymbol{\xi}_{k}$\ is due to
nonuniformity, and affects $E_{k}(\boldsymbol{\xi})\equiv E_{k}%
(\boldsymbol{\xi}_{k})$ as shown in Fig.1 in main text. This completes the
proof of the theorem.
\end{proof}

It should be clear from the above mechanical proof of the theorem that any
time-dependent $\mathfrak{m}_{k}:(\mathfrak{m}_{k}^{\text{s}},\mathfrak{m}%
_{k}^{\text{u}})$ must be uniform and stationary only at the equilibrium point
$\mathfrak{m}_{k\text{eq}}:(\mathfrak{m}_{k\text{seq}},\mathfrak{m}%
_{k\text{ueq}})$ for $\Sigma$. Away from it, $\mathfrak{m}_{k}$ must be nonuniform.

\textsc{(D) The First Law and General Thermodynamics (Gen-Th) for }$\Sigma$
\textsc{: }So far, we have only considered the step (S1), see main text. We
now take the step (S2) and introduce stochasticity by considering an ensemble
of $\bar{\Sigma}$ in which $\mathfrak{m}_{k}$ appears with probability $p_{k}%
$, but without invoking Mec-EQ-P discussed in (C). It should be stressed that
the deterministic Hamiltonian $\mathcal{H}(\left.  \mathbf{x}\right\vert t)$
is oblivious to $p_{k}$ so $\mathfrak{m}_{k}(t)$ is also oblivious to $p_{k}$
in that all its microquantities such as $E_{k}(t)$ do not change with $p_{k}$.
Despite this, they together define various macrostates $\mathfrak{M}:\left\{
\mathfrak{m}_{k},p_{k}\right\}  $ for the same set $\left\{  \mathfrak{m}%
_{k}\right\}  $\ but different sets $\left\{  p_{k}\right\}  $. We put no
restrictions on possible sets $\left\{  p_{k}\right\}  $ so the resulting
macrostates may have nothing to do with what we encounter in classical
thermodynamics (Cl-Th), see the table above, and their thermodynamics, to be
called \emph{general thermodynamics }(Gen-Th), may not always satisfy the
second law $(dS\geq0)$ that is a fundamental axiom (or assumption) in Cl-Th
\cite{CallenS} as we will see below. By taking $\left\{  p_{k}\right\}  $
arbitrary will allow us to determine the root cause of the second law of Cl-Th
and what will cause its violation in the \emph{violation thermodynamics}
(Viol-Th), both of them are contained in Gen-Th.

The macroenergy $E$ of $\bar{\Sigma}$ in Gen-Th is given by the ensemble
average $\left\langle E\right\rangle $%
\[
E=\left\langle E\right\rangle \doteq%
{\textstyle\sum\nolimits_{k}}
p_{k}E_{k},
\]
which is valid for any $N\geq1$. It follows from this that%
\[
dE=%
{\textstyle\sum\nolimits_{k}}
E_{k}dp_{k}+%
{\textstyle\sum\nolimits_{k}}
p_{k}dE_{k}.
\]
From Eq. (\ref{GenMicroforce}) in (A), we observe that the second sum gives
the negative of the SI-macrowork $dW$%
\begin{subequations}
\begin{equation}
dW\doteq%
{\textstyle\sum\nolimits_{k}}
p_{k}dW_{k} \label{generalized work}%
\end{equation}
as the ensemble average of SI-microwork $dW_{k}$. We identify the first sum,
which is $dE+dW$, with SI-macroheat $dQ$ as the ensemble average of
SI-microheat $dQ_{k}$%
\begin{equation}
dQ\doteq%
{\textstyle\sum\nolimits_{k}}
p_{k}dQ_{k}\doteq%
{\textstyle\sum\nolimits_{k}}
p_{k}E_{k}d\eta_{k}, \label{generalized heat}%
\end{equation}
where $\eta_{k}\doteq\ln p_{k}$ is the Gibbs \emph{probability index}
\cite{GibbsS}. It should be stressed that the above identification of
$dQ$\ does not require any size restriction on $\Sigma$\ and any notion of
temperature or entropy; see below. We thus obtain the statement of the
\emph{first law}
\end{subequations}
\begin{equation}
dE=dQ-dW \label{ArbitraryFirstLaw}%
\end{equation}
for $\bar{\Sigma}$ in terms of $dQ$ and $dW$, both defined for any $N$. This
completes the step (S3). It should be clear that $dQ$ and $dW$ are the primary
concepts in the first law and Gen-Th, making them equivalent as both are
independent of SL ($dS\geq0$).

We observe that $dE=0$ in $\bar{\Sigma}$ so that%
\[
E=%
{\textstyle\sum\nolimits_{k}}
p_{k}(t)E_{k}(t)\equiv%
{\textstyle\sum\nolimits_{k}}
p_{k}^{\text{eq}}E_{k}^{\text{eq}}=E_{\text{eq}},
\]
where $p_{k}^{\text{eq}}$\ is the probability of $\mathfrak{m}_{k\text{eq}}$,
and $E_{\text{eq}}$ is the macroenergy of $\mathfrak{M}_{\text{eq}}$. We use
this fact in the first law in Eq. (\ref{ArbitraryFirstLaw}) to obtain a simple
but very remarkable and extremely profound result%
\begin{equation}
dQ=dW\gtreqqless0. \label{DifferentialHeat-Work}%
\end{equation}
which provides the \emph{mechanical formulation }by $dW$ of the stochasticity
inherent in the thermodynamic process through $dQ$. We conclude that the
stochasticity, \textit{i.e.}, the change $\left\{  dp_{k}\right\}  $ in $dQ$
is strongly constrained by the mechanical work $dW$ and its signature for any
$\mathcal{P}$. Because of this constraint, we call Eq.
(\ref{DifferentialHeat-Work}) the \emph{irreversibility principle} (Irr-P)
under steps (S1-S3), which is valid for an isolated $\Sigma$ of any size, and
is a direct consequence of the first law.

We show in Lemma 1 in the main text that the inequality $dW\geq0$ is satisfied
in any \emph{spontaneous} infinitesimal process controlled by Mec-EQ-P. As a
consequence, $dW<0$ must only happen in any \emph{nonspontaneous}
infinitesimal process that violates Mec-EQ-P.

This completes the discussion of Gen-Th and the first law in it.

As $p_{k}$ remains constant in $dW$, it represents an isentropic macroquantity
to justify it as a \emph{mechanical} quantity, the average of the change
$dW_{k}=-dE_{k}$; see Eq. (\ref{generalized work}). On the other hand, $p_{k}$
does not remain constant in $dQ$ so it justifies $dQ$ as a \emph{stochastic}
quantity undergoing entropy change $dS$ but the two are not simply related in
all cases as discussed below.

We are interested in the microwork done during the $\mathfrak{m}_{k\text{eq}}%
$-controlled evolution of $\mathfrak{m}_{k}(t)$ along a trajectory $\gamma
_{k}$ in a thermodynamic process $\mathcal{P}$ in $\mathfrak{S}_{\mathbf{X}}$.
The accumulated microwork along $\gamma_{k}$ in step (S1) follows from Eq.
(\ref{Micropower-work}), and is given by
\begin{equation}
\Delta W_{k}=-\Delta E_{k}\label{IsolatedMicrowork}%
\end{equation}
by the microenergy change $\Delta E_{k}$ along $\gamma_{k}$; recall Eq.
(\ref{Work-kE}). In step (S2), its ensemble average over $\mathcal{P}$ yields
the SI-macrowork%
\begin{subequations}
\begin{equation}
\Delta W=\int_{\mathcal{P}}dW\label{AccumulatedHeat-Work1}%
\end{equation}
in Gen-Th; it remains valid even if $p_{k}$ remains constant over $\gamma_{k}%
$. In the cumulative form of Eq. (\ref{DifferentialHeat-Work}), we have the
identity
\end{subequations}
\begin{subequations}
\begin{equation}
\Delta Q=\Delta W,\label{IntegralHeat-Work}%
\end{equation}
where $\Delta Q$ is given by
\begin{equation}
\Delta Q=%
{\textstyle\sum\nolimits_{k}}
\int_{\gamma_{k}}p_{k}E_{k}d\eta_{k},\label{AccumulatedHeat-Work2}%
\end{equation}
and exists if and only if $p_{k}$ does not remain constant during
$\mathcal{P}$. This is consistent with the conventional wisdom that the
concept of heat does not apply to mechanical bodies for which $\Delta
Q\equiv0$. While $\Delta W$ is determined by the instantaneous value of
$p_{k}$ along $\gamma_{k}$, it does not determine how it changes along it. The
latter is determined by $\gamma_{k}\in\mathcal{P}$ and determines $\Delta Q$.

The macroworks $dW$ and $\Delta W$ play an important role in understanding the
function $E_{\text{w}}$ shown in Fig. 1. We first introduce the cumulative
macrowork%
\end{subequations}
\begin{subequations}
\begin{equation}
W_{\text{eq}}^{\prime}\doteq\int_{\mathcal{P}_{\text{eq}}^{\prime}}%
dW\equiv<\Delta W>_{\mathcal{P}_{\text{eq}}^{\prime}}\label{DeltaW'}%
\end{equation}
along a process $\mathcal{P}_{\text{eq}}^{\prime}$\ starting from EQ point
$\mathfrak{M}_{\text{eq}}$\ to some point $\mathfrak{M}$\ on solid curves in
Fig. 1. We immediately see that $\mathcal{P}_{\text{eq}}^{\prime}$ is in the
direction opposite to the blue arrow for $\mathfrak{M}^{\text{s}}$ but in the
direction of the red arrow for $\mathfrak{M}^{\text{u}}$. We use it to
introduce a \emph{macrowork function}
\begin{equation}
E_{\text{w}}\doteq E_{\text{eq}}-W_{\text{eq}}^{\prime}\label{E_w}%
\end{equation}
of $t$ or $\boldsymbol{\xi}$. It is this average $E_{\text{w}}$\ that is shown
in Fig. 1, and differs from $E=E_{\text{eq}}$ by the cumulative macrowork
$W_{\text{eq}}^{\prime}$. We see that
\begin{equation}
dE_{\text{w}}=-dW,\label{dE_w}%
\end{equation}
a result that provides an equivalent justification of GSL, which is
thermodynamically more intuitive regarding the role of the extremum of
$E_{\text{w}}$ in the evolution of $\mathfrak{M}$. We have focused on
microstates to emphasize the role of mechanics for deriving GSL.

Different components of $\boldsymbol{\xi}_{k}$\ take different times when they
vanish. They are called relaxation times so the components can be ordered
according to them \cite{GujratiS-Hierarchy} as fast and slow. As
$\Sigma_{\text{E}}$ or $\Sigma_{\text{M}}$ evolves in time, they appear or
disappear at different times, making $\mathfrak{m}_{k}$ more or less uniform,
a point that is discussed in (C). For a single $\xi$ in the figure in the main
text, $\mathfrak{M}^{\text{s}}$\ ($\mathfrak{M}^{\text{u}}$) becomes more and
more (less and less) uniform as $t$ increases so that $dS>$ ($<$) $0$ during
any infinitesimal change $dt$.

\textsc{(E) Restriction Thermodynamics (Rest-Th): }We now put a
\emph{particular} form of restriction on $\left\{  p_{k}\right\}  $ for
$\mathfrak{M}$ in $\mathfrak{S}_{\mathbf{Z}}$ of $\Sigma_{\text{M}}$\ that
ensures that its $E$ becomes a \emph{state function} of its $S$ and
$\boldsymbol{\xi}_{k}$, which in turn means that $S$ is a state function of
$E$ and $\boldsymbol{\xi}_{k}$ in $\mathfrak{S}_{\mathbf{Z}}$, whereas it is
not a state function for $\Sigma_{\text{E}}$. Using $E$ as a state function
$E(S,\xi_{\text{E}},\xi_{\text{V}}),N$ and $V$ fixed, we obtain the Gibbs
fundamental relation%
\end{subequations}
\[
dE=TdS-dW,
\]
where $T\doteq\left(  \partial E/\partial S\right)  _{\boldsymbol{\xi}}$\ is
the thermodynamic temperature of $\Sigma_{\text{M}}$, $F_{\text{E}}%
\doteq-\left(  \partial E/\partial\xi_{\text{E}}\right)  _{S,\xi_{\text{V}}%
},F_{\text{V}}\doteq-\left(  \partial E/\partial\xi_{\text{V}}\right)
_{S,\xi_{\text{E}}}$\ are SI-macroforces, and $dW=F_{\text{E}}d\xi_{\text{E}%
}+F_{\text{V}}d\xi_{\text{V}}$. It is the first law now in terms of $dS$ and
$dW$ and yields the \emph{restriction thermodynamics} (Rest-Th), which is
nothing but Gen-Th for $\Sigma_{\text{M}}$. The first term is nothing but $dQ$
in Eq. (\ref{ArbitraryFirstLaw}), which immediately establishes
\begin{equation}
dQ=TdS,T\doteq\left(  \partial E/\partial S\right)  _{\boldsymbol{\xi}}.
\label{dQ-dS}%
\end{equation}
It should be noted that the relationship \cite{GujratiS-EntropyReview} is
simply a the mathematical consequence of the state function
$E(S,\boldsymbol{\xi}_{k})$. However, the functional dependence
$E(S,\boldsymbol{\xi}_{k})$ still does not enforce SL ($TdS\geq0$) as was also
the case for Gen-Th, so Rest-Th allows us to make direct connection with Cl-Th
and Viol-Th. the justification for keeping $T$ in the new formulation
$TdS\geq0$ will be justified in the main text within the context of
generalized SL\ (GSL). We thus note that $dQ$ and $dS$ have the same signature
for positive $T$, but opposite signature for negative $T$.

From Eq. (\ref{dQ-dS}) follows a simple but very remarkable and extremely
profound relation%
\begin{equation}
\left.  dQ\right\vert _{E}=TdS=\left.  dW\right\vert _{E}
\label{MechanicalEntropyFormulation}%
\end{equation}
for $\Sigma_{\text{M}}$\ in Gen-Th (we add an extra suffix $E$ as a reminder
that $E$ is fixed, and should not be confused the suffix in $\xi_{\text{E}}$)
so it remains valid both in Cl-Th and Viol-Th for $\Sigma_{\text{M}}$:%
\begin{equation}%
\begin{array}
[c]{c}%
TdS\geq0\text{ }\Longrightarrow\text{ }\left.  dW\right\vert _{E}\geq0,\\
TdS<0\text{ }\Longrightarrow\left.  dW\right\vert _{E}<0.
\end{array}
\label{EntropyGeneration-Destruction}%
\end{equation}
But what is most remarkable about Eq. (\ref{MechanicalEntropyFormulation}) is
that it provides a purely \emph{mechanical definition} of stochastic entropy
change by using the first law that holds in all kinds of thermodynamics
discussed here. In addition, we also see from this equation that
$T\doteq\left(  \partial E/\partial S\right)  _{\boldsymbol{\xi}}$ for
$\Sigma_{\text{M}}$\ here is no different than $T$ in Eq. (\ref{Temp}) for
$\Sigma$ of any size so we use the latter as the definition of $T$ for any
system of any size.

This completes the introduction of the Rest-Th (we remove the extra suffix $E$
in Eq. (\ref{EntropyGeneration-Destruction}) for simplicity), which is applied
below to a macrostate $\mathfrak{M}$ of $\Sigma_{\text{M}}$\ used for Eq.
(\ref{IndependentCombinations2}), in which $\xi_{k\text{E}}$ must be replaced
by its ensemble average $\xi_{\text{E}}\doteq E_{1}/n_{1}-E_{2}/n_{2}$. A
simple calculation using the state function $S(E,\boldsymbol{\xi}_{k})$ and
$dE=0$ yields%
\begin{equation}
dS=n_{1}n_{2}[(\beta_{1}-\beta_{2})d\xi_{\text{E}}+(\beta_{1}P_{1}-\beta
_{2}P_{2})d\xi_{\text{V}}], \label{EntropyGeneration}%
\end{equation}
where $\beta_{1},P_{1}$, and $\beta_{2},P_{2}$ are the inverse temperature and
pressure of $\Sigma_{1\text{M}}$\ and $\Sigma_{2\text{M}}$, respectively, and
$\beta\doteq1/T$ $=n_{1}\beta_{1}+n_{2}\beta_{2}$ is the inverse temperature
of $\Sigma$. The two terms on the right side in Eq. (\ref{EntropyGeneration})
represent entropic contributions due to the two internal variables
$\xi_{\text{E}}$ and $\xi_{\text{V}}$ in Rest-Th; each must be nonnegative for
SL or negative for its violation. In terms of $dE_{1}=-dE_{2}$ as the
macroenergy change and $dV_{1}=-dV_{2}$ as the volume change of $\Sigma
_{1\text{M}}$, we have
\[
d\xi_{\text{E}}=\frac{1}{n_{1}n_{2}}dE_{1},d\xi_{\text{V}}=\frac{1}{n_{1}%
n_{2}}dV_{1}.
\]
Using $\gamma\doteq-1/T$ introduced by Ramsey \cite{RamseyS}, whose numerical
values define the "hotness" of $\Sigma_{\text{M}}$ as it increases from
$-\infty$ to $+\infty$ covering positive and negative temperatures, Eq.
(\ref{EntropyGeneration}) becomes\qquad\qquad\qquad\qquad\qquad\qquad\qquad%
\begin{equation}
dQ=TdS=-(\Delta\gamma/\gamma)dE_{1}-(\Delta(\gamma P)/\gamma)dV_{1}
\label{EntropyGeneration0}%
\end{equation}
for $\Sigma_{\text{M}}$ in Rest-Th, where $\Delta\gamma\doteq(\gamma
_{2}-\gamma_{1})$, and $\Delta(\gamma P)\doteq(\gamma_{2}P_{2}-\gamma_{1}%
P_{1})$; $\Delta\gamma>0$ means that $\Sigma_{2\text{M}}$ is hotter than
$\Sigma_{1\text{M}}$, and vice-versa. For GSL to hold, we require $dS$ in
$dQ=TdS$\ to be $\geq0$ for $T>0$ $\left(  \gamma<0\right)  $ and $<0$ for
$T<0$ $\left(  \gamma>0\right)  $. For GSL violation, we require $dS<0$ for
$T>0$ and \ $dS>0$ for $T<0$. Both situations are considered in the main text.
Let us consider just the first term above to be specific by setting $dV_{1}%
=0$. We consider various scenarios; (a) and (b) refer to $\mathfrak{M}%
^{\text{s}}$ having $T>0$, and (c) and (d) refer to $\mathfrak{M}^{\text{u}}$
having $T<0$. The analysis here is more extensive compared to an earlier
preliminary and incomplete investigation \cite{GujratiS-Foundations}, where
the issue of catastrophic evaluation was first discussed.

\begin{enumerate}
\item[(a)] For $dQ=dW>0$ and $T>0$, we must have $dE_{1}>0$ for $\mathfrak{M}%
^{\text{s}}$ so that macroheat flows from hot to cold as expected in which
it\ converges to $\mathfrak{M}_{\text{seq}}$ due to an \emph{attractive}
SI-macroforce $\mathbf{F}^{\text{seq}}$ pointing towards SEQ. The SI-evolution
of $\mathfrak{M}^{\text{s}}$ is spontaneous due to $\mathbf{F}^{\text{seq}}$
pointing towards its sink $\mathfrak{M}_{\text{seq}}$. Therefore, as expected
in this case, $dS>0$ so Cl-Th and Gen--GSL-Th remain valid. This is the most
common situation.

\item[(b)] For $dQ=dW<0$ and $T>0$, $dE_{1}<0$ so that macroheat flows from
cold to hot, and $\mathfrak{M}^{\text{s}}$ runs away from $\mathfrak{M}%
_{\text{seq}}$ due to some \emph{repulsive} macroforce $\mathbf{F}%
_{\text{repu}}^{\text{s}}$, distinct from the SI-macroforce $\mathbf{F}%
^{\text{seq}}$,\ to eventually\ converge to $\mathfrak{M}_{\text{cata}%
}^{\text{s}}$ by becoming more and more nonuniform. The evolution of
$\mathfrak{M}^{\text{s}}$ is not spontaneous as is in (a) and $\mathfrak{M}%
_{\text{seq}}$ is no longer the sink. In this case, $dS<0$ (Viol-Th), but we
also violate GSL (Gen--GSL-Th). \ \ 

\item[(c)] For $dQ=dW>0$ but $T<0$ for $\mathfrak{M}^{\text{u}}$, $dE_{1}<0$
so that macroheat flows from cold to hot and $\Sigma_{\text{M}}$ becomes more
and more nonuniform because of the instability in it as discussed in the main
text. The spontaneous evolution of $\mathfrak{M}^{\text{u}}$\ from its source
$\mathfrak{M}_{\text{ueq}}$ under the \emph{repulsive} SI-macroforce
$\mathbf{F}^{\text{ueq}}$\ is catastrophic in that it converges to a
catastrophic macrostate $\mathfrak{M}_{\text{cata}}^{\text{u}}$. In this case,
$dS<0$ so SL is violated (Viol-Th) but GSL (Gen--GSL-Th) remains valid.

\item[(d)] For $dQ=dW<0$ and $T<0$, we must have $dE_{1}>0$ so that macroheat
flows from hot to cold in $\mathfrak{M}^{\text{u}}$. In this case,
$\mathfrak{M}^{\text{u}}$ nonspontaneously converges to $\mathfrak{M}%
_{\text{ueq}}$ due to an \emph{attractive} macroforce, which is distinct from
the repulsive SI-macroforce $\mathbf{F}^{\text{ueq}}$, with$\ dS>0$ so Cl-Th
remains valid\ but not Gen--GSL-Th.\ Thus, $\mathfrak{M}_{\text{ueq}}$ is no
longer the source for $\mathfrak{M}^{\text{u}}$-evolution.
\end{enumerate}

We analyze nonspontaneous processes further. We first consider (b). As
$\mathfrak{M}^{\text{s}}$ runs away from its sink $\mathfrak{M}_{\text{seq}}$
to a new macrostate $\mathfrak{M}^{\text{s}\prime}$, the latter further runs
away from $\mathfrak{M}^{\text{s}}$ by its $\xi_{\text{E}}$\ deviating further
from their values in $\mathfrak{M}^{\text{s}}$. Thus, we get successive
macrostates $\mathfrak{M}^{\text{s}(p)},p=0,1,2,\cdots$ which run away from
the sink farther and farther, during which $dS^{(p)}$ remains non-positive as
$p$ increases. Therefore, the evolution to $\mathfrak{M}_{\text{cata}%
}^{\text{s}}$ is \emph{catastrophic} in that it makes $\Sigma_{\text{M}}$
highly nonuniform due to unexplained nonsystem repelling macroforce
$\mathbf{F}_{\text{repu}}^{\text{s}}$ that mutilates $\mathfrak{M}%
_{\text{seq}}$. It follows from the stability of $\mathfrak{M}^{\text{s}}$
considered here that $\Sigma_{1\text{M}}$\ and $\Sigma_{2\text{M}}$ are also
stable so their specific heats at constant volume are nonnegative. As
$dE_{1}=-dE_{2}<0$, the disparity $\Delta\gamma$ continues to increase with
$\Sigma_{2\text{M}}$ getting more hot and $\Sigma_{1\text{M}}$ getting more
cold, until $\Delta\gamma\ $takes its maximum value $\Delta\gamma
_{\text{cata}}$ in $\mathfrak{M}_{\text{cata}}^{\text{s}}$. We now consider
(d), where a similar discussion can also carried out for $\mathfrak{M}%
^{\text{u}}$ but different conclusions. Here, $\mathfrak{M}^{\text{u}}$ gets
more uniform but the uniformity is not due to SI-macroforce so the evolution
is not governed by SL. On the other hand, the catastrophe in (c) occurs
spontaneously at negative $T$ so $dS<0$ should be treated as being governed by
SL. Thus, GSL seems to capture spontaneous processes at both positive and
negative temperatures. In this sense, GSL subsumes SL.

The discussion is easily extended to include the second term in Eq.
(\ref{EntropyGeneration0}) with same conclusions, which is consistent with
Lemma in the main text.

We clearly see from (a) and (d) that $dS>0$ is not always a consequence of
spontaneous processes. In (d), it is a consequence not only of negative $T$
but also of negative $dW$ performed by nonsystem forces that result in a
nonspontaneous process. This process is not controlled by SL so $dS>0$ has no
significance.\ Similarly, $dS<0$\ in (b) and (c) shows that it is not always a
consequence of nonspontaneous processes. In (b), it is a consequence only of
negative $dW$ performed by nonsystem forces that result in a nonspontaneous
process such as during the creation of internal constraints as explained by
Callen \cite{CallenS}. Again, this process is not controlled by SL, while
SI-macrowork $dW>0$\ or the removal of the internal constraint is controlled
by SL so $dS<0$ in (d) has no significance for SL-violation.\ From (a) and
(c), we observe that GSL is always a consequence of spontaneous processes, but
fails for nonspontaneous processes in (b) and (c).\


\begin{thebibliography}{99}
\bibitem {Internal}All processes in $\Sigma$ are synonymously identified as
spontaneous, internal, or SI processes in that none of them is affected by
anything exterior to it. Thus, we use these terms interchangeably.
\bibitem {Note-stochasticity}As constant $p_{k}$ refer to a purely mechanical
system, stochasticity for us will always mean changing $p_{k}$ so that it will
always refer to nonvanishing $\left\{  dp_{k}\right\}  $ and nonvanishing $dQ$
as discussed in the text. We use dissipation to always mean converting
$dW\geq0$ to $dQ\geq0$, but use irreversibility to imply any signature of $dW$.
\bibitem {NoEngine}The formulations of SL by Clausius and by Kelvin-Planck are
not applicable to isolated systems as they cannot perform any cyclic operation
of an engine.
\bibitem {Equilibrium}The EQ macrostate $\mathfrak{M}_{\text{eq}}$ of $\Sigma$
is identified by its entropy having the maximum value for fixed observable
$\mathbf{X}$ in the state space $\mathfrak{S}_{\mathbf{X}}$ as defined in the
text later.
\bibitem {FDNote}To paraphrase Kubo \cite{Kubo}, the internal relationship
between the frictional force in Brownian motion resulting in $dQ$ and the
random microscopic driving forces resulting in $\left\{  dW_{k}\right\}  $ is,
" ... in fact, a very general matter, which is manifested in the so-called
fluctuation-dissipation theorem ... because both come from the same origin."
\bibitem {Note-StochasticThermodynamics}The concepts of generalized macrowork
and macroheat are very different from those considered in \cite{Siefert} and
\cite{Sekimoto}, where these concepts refer to exchange microquantities with a
medium, which is absent for our isolated system $\Sigma$, so they vanish. Our
concepts are intrinsic to $\Sigma$ and are usually nonzero. Despite this, the
generalized macrowork according to Irr-P is completely converted into
generalized macroheat but not the other way around as our microscopic approach
will justify. This makes $dW$ the primary concept and $dQ$ the secondary concept.
\bibitem {Ramsey1}The assumption for $\mathfrak{m}_{k}$ is the extension of
the assumption by Ramsey \cite{Ramsey} for $\mathfrak{M}_{\text{eq}}$ to NEQ situation.
\end{thebibliography}
\begin{thebibliography}{9999}
\bibitem {micro-macro}The prefix micro- is used for quantities pertaining to
microstates, while macro- is used for quantities pertaining to macrostates.
\end{thebibliography}
\end{document}